\documentclass[conference]{IEEEtran}
\IEEEoverridecommandlockouts

\usepackage[utf8]{inputenc}
\usepackage[T1]{fontenc}
\usepackage{mathtools}
\usepackage{amssymb}
\usepackage{amsthm}
\usepackage{xspace}
\usepackage{float}
\usepackage{pgf}
\usepackage{tikz}
\usepackage{array}
\usepackage{booktabs}
\usepackage{changepage}
\usepackage{hyperref}
\usepackage{todonotes}
\usepackage{enumitem}
\usepackage{cite}
\usepackage{xcolor}
\definecolor{light-gray}{gray}{0.75}
\definecolor{darkgreen}{RGB}{0,200,0}
\definecolor{lightpurple}{RGB}{220,0,220}
\usepackage{cleveref}

\usetikzlibrary{arrows,calc,automata,shapes,positioning}
\tikzset{AUT style/.style={>=angle 60,initial text= ,every edge/.append style={thick},every state/.style={thick,minimum size=20,inner sep=0.5},states/.style={state,ellipse,minimum height=40,minimum width=30}}}

\newtheorem{theorem}{Theorem}[section]
\newtheorem{lemma}[theorem]{Lemma}
\newtheorem{corollary}[theorem]{Corollary}
\newtheorem{proposition}[theorem]{Proposition}

\newenvironment{proofsketch}{\noindent\emph{Proof sketch.}}{\hfill$\square$}

\theoremstyle{definition}
\newtheorem{definition}[theorem]{Definition}
\newtheorem{remark}{Remark}
\newtheorem{example}{Example}

\newcommand{\cceq}{\mathop{::=}}

\renewcommand{\epsilon}{\varepsilon}
\renewcommand{\phi}{\varphi}

\newcommand{\pow}[1]{2^{#1}}
\newcommand{\nats}{\mathbb{N}}
\newcommand{\size}[1]{|#1|}
\newcommand{\set}[1]{\{ #1 \}}
\newcommand{\ap}[0]{\mathrm{AP}}

\newcommand{\ltl}{{LTL}\xspace}
\newcommand{\ctl}{{CTL}\xspace}

\newcommand{\poly}{\textsc{P}\xspace}
\newcommand{\np}{\textsc{NP}\xspace}
\newcommand{\conp}{co\textsc{NP}\xspace}
\newcommand{\sigmatwo}{$\Sigma_2^{\textsc{P}}$\xspace}
\newcommand{\pitwo}{$\Pi_2^{\textsc{P}}$\xspace}
\newcommand{\sigmathree}{$\Sigma_3^{\textsc{P}}$\xspace}
\newcommand{\pspace}{\textsc{PSpace}\xspace}
\newcommand{\expt}{\textsc{ExpTime}\xspace}

\newcommand{\twoexpt}{\textsc{2ExpTime}\xspace}

\newcommand{\sharppoly}{\textsc{\#P}\xspace}
\newcommand{\sharpsigmatwo}{$\textsc{\#P}^{\textsc{NP}}$\xspace}
\newcommand{\sharpsigmathree}{$\textsc{\#P}^{\Sigma_2^{\textsc{P}}}$\xspace}

\newcommand{\kripke}{\mathcal{K}}

\newcommand{\runs}[1]{\mathcal{R}(#1)}
\newcommand{\traces}[1]{\mathcal{L}(#1)}

\newcommand{\game}[1]{\mathcal{G}_{#1}}
\newcommand{\sat}{\mathtt{Sat}}
\newcommand{\unsat}{\mathtt{Unsat}}
\newcommand{\vsat}{S_{\sat}}
\newcommand{\vunsat}{S_{\unsat}}
\newcommand{\val}[1]{val(#1)}
\newcommand{\valbar}[1]{\overline{val}(#1)}
\newcommand{\valtwoturn}[1]{val_{2turn}(#1)}

\newcommand{\valconcurrent}[1]{val_{concur}(#1)}
\newcommand{\importance}[1]{\mathcal{I}(#1)}
\newcommand{\importancebar}[1]{\overline{\mathcal{I}}(#1)}
\newcommand{\importancetwoturn}[1]{\mathcal{I}_{2turn}(#1)}

\newcommand{\importanceconcurrent}[1]{\mathcal{I}_{concur}(#1)}
\newcommand{\modal}{\mathcal{M}}

\newcommand{\Deltamay}{\Delta_{may}}
\newcommand{\Deltamust}{\Delta_{must}}
\newcommand{\sons}[2]{Sons_{#2}(#1)}

\newcommand{\choicessat}[2]{C_{S}(#1,#2)}
\newcommand{\choicesunsat}[2]{C_{U}(#1,#2)}

\def\myscale{0.8}

\newcommand\restrict[2]{{
		\left.\kern-\nulldelimiterspace
		#1
		\vphantom{\big|}
		\right|_{#2}
}}

\title{Responsibility and verification: \\Importance value in temporal logics
	\thanks{This work was funded by DFG grant 389792660 as part of TRR~248 (see \url{https://perspicuous-computing.science}), the Cluster of Excellence EXC 2050/1 (CeTI, project ID 390696704, as part of Germany's Excellence Strategy), DFG-projects BA-1679/11-1 and BA-1679/12-1, the Research Training Group QuantLA (GRK 1763). Kiefer is supported by a Royal Society University Research Fellowship.}}

\author{
	\IEEEauthorblockN{
		Corto Mascle\IEEEauthorrefmark{1},
		Christel Baier\IEEEauthorrefmark{2},
		Florian Funke\IEEEauthorrefmark{2},
		Simon Jantsch\IEEEauthorrefmark{2},
		Stefan Kiefer\IEEEauthorrefmark{3}} 
	\IEEEauthorblockA{\IEEEauthorrefmark{1}ENS Paris-Saclay, France}
	\IEEEauthorblockA{\IEEEauthorrefmark{2}Technische Universität Dresden, Germany}
	\IEEEauthorblockA{\IEEEauthorrefmark{3}University of Oxford, UK}
}

\date{April 2020}

\begin{document}

\IEEEoverridecommandlockouts
\maketitle

\begin{abstract}
	We aim at measuring the influence of the nondeterministic choices of a part of a system on its ability to satisfy a specification. For this purpose, we apply the concept of Shapley values to verification as a means to evaluate how important a part of a system is. The importance of a component is measured by giving its control to an adversary, alone or along with other components, and testing whether the system can still fulfill the specification. We study this idea in the framework of model-checking with various classical types of linear-time specification, and propose several ways to transpose it to branching ones. We also provide tight complexity bounds in almost every case.
	
\end{abstract}


\section{Introduction}

Classical model-checking algorithms try to detect undesired behaviors in a formal system with reference to a given specification, and the system is deemed correct if they cannot find one. However, simply knowing that the system satisfies the specification is in practice often unsatisfactory: we also want to know \emph{why} it does, or does not. Especially in the case that the specification is violated, knowing \emph{where} in the system to look for a potential model repair can significantly reduce troubleshooting times for both engineers and users.

To this end, Chockler, Halpern and Kupferman defined a notion of \emph{causality} aimed at explaining which parts of a system are relevant for the satisfaction of a specification $\phi$~\cite{Causality2008}. More specifically, a state $s$ is considered a cause for $\phi$ with respect to an atomic proposition $p$ if the value of $p$ can be swapped in a subset of the states $T$ such that further swapping the value of $p$ in $s$ turns $\phi$ from being satisfied to being violated (we say that $(s,T)$ is \emph{critical}). Counterfactual reasoning in this spirit (i.e., had the cause not occurred, then the event would not have happened) has a rich history in the philosophy and moral responsibility literature, and has been formalized in the framework of \emph{structural equation models} \cite{HalpernPearl2004,HalpernP05}, on which the work \cite{Causality2008} is based. Causes are further assigned a \emph{degree of responsibility} by taking the inverse of the size of the smallest set $T\cup\{s\}$ such that $(s,T)$ is critical. This numerical value, adapted from \cite{ChocklerHalpern2003}, is designed to measure the impact of the state on the specification: Causes with high degree of responsibility point to small changes of the system that have the power to crucially alter its behavior.

In this paper we define a novel measure for the influence of a state on a specification, called the \emph{importance}. While it is related to the degree of responsibility of~\cite{Causality2008}, a significant difference appears in how the counterfactuality principle is invoked. The degree of responsibility relies on hypothetical modifications of the structure and answers the question ``Is the system still working if the truth value of this atomic proposition in that subset of states is switched?'' In contrast, we never modify the system, but look at how its nondeterministic choices are resolved, thus tackling the question ``Does the system yield a satisfying run if the subset of states is under control (i.e., behaving in a manner conducive to the functioning) while the others are not (i.e., behaving antagonistically)?'' Hence, our definition of importance relies on a new viewpoint of what constitutes a critical pair, based on capturing the specific nondeterministic choices available in the states. 

The approach above determines the impact of a subset of states on the satisfaction of a specification. In order to turn this information into the \emph{individual} importance of a state (or a component) we employ a solution concept from cooperative game theory, called the \emph{Shapley value}~\cite{Shapley1953}. In a context of collaborative multi-agent interaction, Shapley values aim at measuring how beneficial the participation of a specific agent is in reaching some objective. Translated to Kripke structures, the idea is to compute the probability that taking control over a particular state makes the system work as intended, where the control over states is taken in a (uniformly) random order. The importance distills those parts of the system whose choices are \emph{crucial} for its functioning. 

As an example, consider a system testing a server $\mathbf{sv}$ by sending regular requests. If the server does not respond correctly, the system retries to send a request; if it does respond correctly, then the system may wait before testing again. We represent this system by a Kripke structure, displayed on the left in Figure \ref{realistic}. Consider the specification stating that the system should make infinitely many tests and receive only finitely many incorrect answers (modeled by the LTL formula $\phi = GF\,\mathbf{check} \land FG\,\neg \mathbf{fail}$). The system fails this condition if $\mathbf{sv}$ malfunctions and fails infinitely often or if $\mathbf{ok}$ waits indefinitely from some point on without rechecking the server. As the other states cannot enforce breaching $\phi$ without $\mathbf{sv}$ and $\mathbf{ok}$, the importance of these two states is $1/2$ and that of the other states is $0$.

Let us now add a backup server $\mathbf{sv'}$ with the same role as $\mathbf{sv}$ (as displayed on the right in Figure \ref{realistic}). Then the system succeeds if it loops infinitely often between $\mathbf{ok}$, $\mathbf{check}$ and the set $\{\mathbf{sv}, \mathbf{sv'}\}$, which is only possible if $\mathbf{ok}$ and at least two of $\mathbf{check}$, $\mathbf{sv}$ and $\mathbf{sv'}$ behave well. In this case we get an importance of $1/2$ for $\mathbf{ok}$ and  $1/6$ for $\mathbf{check}$, $\mathbf{sv}$ and $\mathbf{sv'}$ (these values are explained in detail in \Cref{IntroExample}). This is a numerical interpretation of the fact that control over the behavior of $\mathbf{ok}$ is more critical to the functioning of the system: unfortunate choices made in $\mathbf{ok}$ (i.e., avoiding further tests forever) instantly make the system fail. The equal importance of $\mathbf{check}$, $\mathbf{sv}$, and $\mathbf{sv'}$ reflects the fact that -- although their actual roles in the system differ -- they play interchangeable parts when only the functioning is concerned: any two of them are needed to make the system work.

 It is noteworthy that a variant of the degree of responsibility based on our notion of critical pair (and applied in reverse fashion, i.e., from violation of $\phi$ to satisfaction of $\phi$) would not be able to distinguish $\mathbf{check}$, $\mathbf{ok}$, $\mathbf{sv}$, and $\mathbf{sv'}$ as it evaluates to $1/3$ for each of these states. Roughly speaking, the degree of responsibility only takes a \emph{minimal} critical pair into account, whereas the importance computes a weighted average over the size of \emph{all} critical pairs that a state belongs to. The rationale for this is that belonging to many critical pairs makes the state less dependent on behavior outside of its control, and hence more powerful.

\begin{figure}
	\centering
	\begin{tikzpicture}[xscale=1.5,yscale=1.3,AUT style]

\node[scale=\myscale, state] (OK) at (0.5,1) {$\mathbf{ok}$};
\node[state, scale=\myscale,minimum size=1.3cm,initial] (RC) at (1,3) {$\mathbf{check}$};
\node[state, scale=\myscale+0.1] (S1) at (1,2) {$\mathbf{sv}$};
\node[state, scale=\myscale, minimum size=1.cm] (GU) at (1.5,1) {$\mathbf{fail}$};

\node[state, scale=\myscale] (OK0) at (3.5,1) {$\mathbf{ok}$};
\node[state, scale=\myscale,minimum size=1.3cm,initial] (RC0) at (4,3) {$\mathbf{check}$};
\node[state, scale=\myscale+0.1] (S10) at (3.7,2) {$\mathbf{sv}$};
\node[state, scale=\myscale+0.1] (S20) at (4.3,2) {$\mathbf{sv'}$};
\node[state, scale=\myscale, minimum size=1.cm] (GU0) at (4.5,1) {$\mathbf{fail}$};

\path[->, bend left=20] (OK) edge node[right] {} (RC);
\path[->, bend right=20] (GU) edge node[right] {} (RC);
\path[->] (S1) edge node[right] {} (OK);
\path[->, loop left] (OK) edge node[right] {} (OK);
\path[->] (S1) edge node[right] {} (GU);
\path[->] (RC) edge node[right] {} (S1);

\path[->, bend left=30] (OK0) edge node[right] {} (RC0);
\path[->, bend right=30] (GU0) edge node[right] {} (RC0);
\path[->] (S10) edge node[right] {} (OK0);
\path[->, loop left] (OK0) edge node[right] {} (OK0);
\path[->] (S10) edge node[right] {} (GU0);
\path[->] (RC0) edge node[right] {} (S10);
\path[->] (S20) edge node[right] {} (OK0);
\path[->] (S20) edge node[right] {} (GU0);
\path[->] (RC0) edge node[right] {} (S20);

\end{tikzpicture}
	\caption{Two simple systems, as used in the introductory example.}
	\label{realistic}\label{fig:intro}
\end{figure}
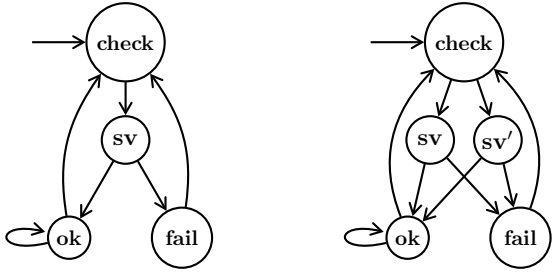

The construction of the importance value as outlined above gives rise to the following three problems, whose complexity we study in this paper for a wide range of specifications.
 The \emph{value problem} consists in determining if a subset of states of a given Kripke structure can guarantee that the specification is respected when the other states act in an adversarial way. The \emph{importance problem} asks for the actual importance value. Finally, the \emph{usefulness problem} asks whether a state of a system has positive importance, i.e., whether its behavior has any influence at all on the satisfaction of the specification. In fact we  define the importance value in the presence of a prescribed partition of the state space, and study the complexity problems in this generalized setting. This allows us to capture more realistic scenarios such as the importance of a \emph{system component} in a composite architecture.

 
Table~\ref{ComplexityTable} summarizes the complexity results obtained throughout the paper. We write $\in \mathcal{C}$ when the problem is in class $\mathcal{C}$ and we do not have a matching lower bound, and just $\mathcal{C}$ when the problem is $\mathcal{C}$-complete. Since our examinations spread over a wide range of specifications, our results crucially rely on a diverse game-theoretic toolkit. 

\renewcommand{\arraystretch}{1.7}

\begin{table*}
\caption{A summary of the results on the complexity of the value, usefulness\hspace{\textwidth} and importance problems for various types of specifications.}
\label{ComplexityTable}
\begingroup
\footnotesize
\begin{tabular}{m{\linewidth*12/100}m{\linewidth*15/100}m{\linewidth*15/100}m{\linewidth*15/100}m{\linewidth*15/100} m {\linewidth*15/100}}
	& \textbf{Büchi} & \textbf{Rabin} & \textbf{Streett} & \textbf{Parity} & \textbf{Explicit Muller} \\
	\textbf{Value}      & \poly~\cite{Immerman1981} & \np~\cite{emerson1988complexity} & \conp~\cite{emerson1988complexity} & $\in$ \np $\cap$ \conp~\cite{emerson1993model} & \poly~\cite{Horn2008}\cite{Immerman1981} \\
	\textbf{Usefulness} & \np  (Prop. \ref{usefulnessReach})    & \sigmatwo (Prop. \ref{usefulnessRabin}) & \sigmatwo (Cor. \ref{usefulnessImportanceStreett}) & \np (Prop. \ref{usefulnessReach})  & \np   (Prop. \ref{usefulnessReach})   \\
	\textbf{Importance} & \sharppoly (Thm. \ref{sharpPcompleteReachability})  & \sharpsigmatwo (Thm. \ref{importanceRabin}) & \sharpsigmatwo (Cor. \ref{usefulnessImportanceStreett}) & \sharppoly (Thm. \ref{sharpPcompleteReachability}) & \sharppoly (Thm. \ref{sharpPcompleteReachability})    \\
\end{tabular}
\\[0.3cm]
\begin{tabular}{m{\linewidth*12/100}m{\linewidth*15/100}m{\linewidth*15/100}m{\linewidth*15/100}m{\linewidth*20/100}}
	&  \textbf{Emerson-Lei} & \textbf{LTL} & \textbf{2-turn CTL}   & \textbf{Concurrent CTL} \\
	\textbf{Value}   & \pspace~\cite{HunterDawar2005} & \twoexpt~\cite{Rosner1992}    & \sigmatwo (Prop. \ref{2turnCTLvalue})  &  $\in$ \expt (Rmk. \ref{NashComplexity})\\
	\textbf{Usefulness} & \pspace (Thm. \ref{EmersonLei})   & \twoexpt (Thm. \ref{2exptimeltlimportance}) & \sigmathree (Prop. \ref{2turnCTLusefulness})   & $\in$ \expt  (Rmk. \ref{NashComplexity})\\
	\textbf{Importance} & \pspace (Thm. \ref{EmersonLei})  & \twoexpt (Thm. \ref{2exptimeltlimportance}) &  \sharpsigmathree (Thm. \ref{2turnCTLimportance}) & $\in$ \expt (Rmk. \ref{NashComplexity}) \\ 
\end{tabular}
\endgroup
\end{table*}

The paper is split into three parts: In the first part we define the notions in the general setup of turn-based two-player games on finite graphs. Then we apply these notions in order to define the importance on Kripke structures with respect to LTL specifications, and finally we look at the case of CTL specifications on modal transition systems. The proofs missing in the main document due to space constraints can be found in the appendix.

\subsection{Related work}

The complexity of computing the aforementioned degree of responsibility was examined for the general class of structural equation models in \cite{ChocklerHalpern2003} and for Boolean circuits in \cite{Causality2008}. They are closely related to the complexity results about deciding causality \cite{EiterL02a,EiterL02b}. 

Our work ties into a ubiquitous quest for powerful \emph{explanations} of model-checking results. If a system satisfies a specification, then \emph{coverage estimation} has been used to analyze which parts of the system are essential for the successful verification result \cite{HoskoteKHZ99,ChocklerKKV01,ChocklerKV01,ChocklerKV06}. As in the definition of the degree of responsibility, the idea is to apply small changes to the system (\emph{mutants}) and check the resulting effect on the specification. \emph{Vacuity detection}, on the other hand, applies the principle of small changes to the specification \cite{BeerBER97,KupfermanV99,PurandareS02}. This strand of research aims at checking whether the specification is satisfied in an undesired, trivial fashion (typically due to insufficient modeling of the system). Coverage and vacuity have been shown to exhibit a formal duality \cite{KupfermanLS08}, and recent work on the subject analyzes network formation games \cite{BielousK20}.

In the case of an unsuccessful verification process, one of the powerful features of many model checking approaches is the ability to generate a counterexample \cite{ClarkeGMZ95}. In order to extract further diagnostic information, there has been extensive work on \emph{localizing} errors in faulty traces \cite{BallNR03,Zeller02,GroveV03,RenieresR03,Groce04,GroceCKS06}. Typically, one compares an erroneous trace with a successful one that lies nearby with respect to a suitable metric. Early detection of error traces has been investigated in \cite{deAlfaroHM00}, where a game-like description close to ours between a system module and its environment has been used. Explaining counterexamples using the notion of causality from~\cite{Causality2008} has been presented in~\cite{BeerBCOT2012}.

The Shapley value is a classical solution concept in economics and has recently received considerable attention in the computer science literature. 
Shapley-like values have been used as explanations for machine learning models, where they estimate the impact of the input parameters on the outcome \cite{LL17,LEL18,SundararajanN20}. They have also been employed as a means by which centrality in networks can be measured \cite{TarkowskiMRW2017} or responsibilities can be assigned in game-like structures \cite{Yazdanpanahetal19}. Computational approaches for the Shapley value are given in  \cite{DengP,FatimaWJ08,SkibskiRMW19,SkibskiMSWY20}. For a variety of recent results and applications of Shapley values we refer to~\cite{AlgabaFS2019}.

\section{Preliminaries}\label{sec:prelims}

\subsection{Words and structures}
\label{StructDefinitions}

\paragraph{Words and trees} Let $A$ be an alphabet. We denote by $A^*$ (resp. $A^\omega$) the set of finite (resp. infinite) words over $A$. Given a word $w$, we write $\size{w}$ for its length and, for all $0 \leq i < \size{w}$, we write $w_i$ for the $(i+1)$th letter of $w$. 

An infinite tree $t$ over $A$ is a prefix-closed subset of $A^*$ such that for all $p \in t$, there exists $a \in A$ such that $pa \in t$. The set of sons of a node $p$ of the tree $t$ is denoted by $\sons{p}{t} = pA \cap t$.
 
\paragraph{Kripke structures} A \emph{Kripke structure} $\kripke$ is a 5-tuple $(S,\ap,\Delta,init,\lambda)$ where $S$ is a finite set of \emph{states}, $\ap$ is a finite set of \emph{atomic propositions}, $\Delta \subseteq S \times S$ is a set of \emph{transitions}, $init$ is an \emph{initial state} and $\lambda : S \to \pow{\ap}$ is a \emph{labeling function}. For every $s \in S$, we define its image under $\Delta$ as $\Delta(s) = \set{t \in S\mid (s,t) \in \Delta}$, and we always assume $\Delta(s)$ to be nonempty for all $s$. A \emph{run} of a Kripke structure $\kripke=(S,\ap,\Delta,init,\lambda)$ is an infinite sequence $r \in S^\omega$ such that $r_0 = init$ and for all $i \in \nats$, we have $(r_i, r_{i+1}) \in \Delta$. To every run $r$ we can associate a \emph{trace}, which is the sequence of labelings $\lambda(r_0) \lambda(r_1) \cdots$. The set of runs of $\kripke$ is denoted by $\runs{\kripke}$ while the set of traces it generates is called $\traces{\kripke}$.

\paragraph{Modal transition systems}\label{MTSdefinition} A \emph{modal transition system} (MTS)~\cite{LarsenT1988} $\modal$ is a 6-tuple $(S,\ap,\Deltamay,\Deltamust,init,\lambda)$ where $S$ is a finite set of states, $\ap$ is a set of atomic propositions, $\Deltamust, \Deltamay \subseteq S\times S$ are sets of transitions such that $\Deltamust \subseteq \Deltamay$, $init \in S$ is an initial state and $\lambda : S \to \pow{\ap}$ is a labeling function. We assume $\Deltamust(s)$ to be nonempty for every state $s$. 
We call a Kripke structure $\kripke = (S,AP,\Delta,init,\lambda)$ an \emph{implementation} of $\modal$ if $\Deltamust \subseteq \Delta \subseteq \Deltamay$.
This is in contrast to other works in the modal transition system literature which usually consider a more general notion of implementation based on refinement relations (see~\cite{Kretinsky2017} for a recent overview).

\subsection{Temporal logics} 

We now define the syntax of the two logics we will consider in this paper, LTL and CTL.
For the semantics and basic properties of these logics we refer the reader to~\cite{ClarkeGP2001} or~\cite{BaierKatoen2008}.

\paragraph{Linear temporal logic} The formulas of LTL are given by the grammar
\[
\phi {} \cceq {}  a \mid \phi \lor \phi \mid \neg \phi \mid X \phi \mid \phi U \phi
\]
with $a$ ranging over a finite set of atomic propositions $AP$.

LTL formulas are evaluated on infinite words over $\pow{\ap}$.
We extend the set of operators with $\top, \bot, \land, F, G$ and $R$ in the usual way.

\paragraph{Computation tree logic} The syntax of CTL is defined by the grammar
\[
\phi {} \cceq {}  a \mid \phi \lor \phi \mid \neg \phi \mid EX \phi \mid E \phi U \phi \mid A \phi U \phi
\]
with $a$ ranging over a finite set of atomic propositions $\ap$.

CTL formulas are evaluated on infinite trees over $\pow{\ap}$.
We extend the set of operators with $\top, \bot, \land, EF, EG, ER, AX, AR, AF$ and $AG$ in the usual way. 

\subsection{Games on graphs} A \emph{directed graph} $G$ is a pair $(V,E)$ with $V$ a set of \emph{vertices} and $E \subseteq V^2$ a set of \emph{edges}. An \emph{arena} is a tuple $(G,\vsat, \vunsat)$ with $G$ a graph and $\vsat, \vunsat$ a partition of its vertices. We say that the vertices of $\vsat$ \emph{belong} to player $\sat$, or are \emph{controlled} by player $\sat$ (and similarly for $\unsat$). 
As the graphs we consider will be induced by Kripke structures, we will use from now on $S$ to denote the set of vertices and $\Delta$ for the edges. We will also refer to vertices as states and edges as transitions.

A \emph{game} $\game{}$ is defined by an arena $((S,\Delta),\vsat, \vunsat)$, an initial vertex $init \in S$ and a \emph{winning condition} (also called \emph{objective}) $\Omega \subseteq S^{\omega}$. 
A \emph{play} of $\game{}$ is an infinite sequence $p \in S^\omega$ such that $p_0 = init$ and for all $i \in \nats$, $(p_i, p_{i+1}) \in \Delta$.
A \emph{strategy} for $\sat$ (resp. $\unsat$) is a function $\sigma_{\sat} : S^*\vsat \to S$ (resp. $\sigma_{\unsat} : S^*\vunsat \to S$).
A play $p$ is said to \emph{respect} a strategy $\sigma$ of $\sat$ (resp. $\sigma$ of $\unsat$) if for all $i \in \nats$, if $p_i \in \vsat$ (resp. $p_i \in \vunsat$) then $p_{i+1}=\sigma(p_0\ldots p_i)$.
A strategy $\sigma$ of $\sat$ (resp. $\unsat$) is \emph{winning} if for every play $p$ respecting $\sigma$ we have $p \in \Omega$ (resp. $p\notin \Omega$). 
A game $\game{}$ is \emph{determined} if there exists a winning strategy for either $\sat$ or $\unsat$.

For more information on infinite games played on finite graphs we refer to~\cite{GradelTW2001}. In particular, we will use several classical winning conditions on such games, whose definitions can be found in \cite[Chapter 2]{GradelTW2001}. 

\subsection{Complexity classes} We consider mostly well-known and classical complexity classes, a description of which can be found, e.g., in~\cite{Papadimitriou2007}. We use logarithmic space reductions for the decision problems and Turing reductions for the counting complexity classes.

\section{A general definition of importance in two-player games}

Let $\game{}$ be a two-player game between $\sat$ and $\unsat$ on an arena $((S,\Delta),\vsat , \vunsat)$. Let $\Omega \subseteq S^\omega$ be $\sat$'s objective (i.e. the set of plays of $\game{}$ she wins). In order for the game to be determined, we assume $\Omega$ to be a Borel set.

We start by defining a general notion of \emph{importance} of a state (or a set of states, in a given partition), which is a measure of how much a state contributes towards $\sat$ winning the game.
In other words, if $\sat$ is restricted to controlling only some of her states (for example, due to resource constraints) she should opt to control the ones with high importance in order to win the game.

\begin{definition}
 For all sets of states $\vsat' \subseteq S$, we define $\game{\vsat'}$ as the game between $\sat$ and $\unsat$ played on the arena $((S,\Delta), \vsat',S\setminus \vsat')$  with the same initial state and the same objective $\Omega$ for $\sat$. 
\end{definition}

\begin{definition}[Value of a state subset]
\label{defvalue}
For all sets of states $\vsat' \subseteq S$, we define the value of the set $\vsat'$ as \[
\val{\vsat'} =
\begin{cases}
1 \text{ if } \sat \text{ has a winning strategy for } \game{\vsat'}\\
0 \text{ if } \unsat \text{ has a winning strategy for } \game{\vsat'}
\end{cases}
\]
Note that the value is defined with respect to a game, but that game does not appear in the notation as we will always make it clear from context. The value is well-defined for all $\vsat'$ as we assumed the objective to be a Borel set, thus the game is determined.
\end{definition}

With the definition of \emph{value} of a subset of states, we are in the position of defining the importance of a state.
This definition corresponds to the classical formula for the Shapley value~\cite{Shapley1953}.
In our context it can be explained as follows: for a given state $s$, it counts the number of orderings of the states in $\vsat$ such that if $\sat$ gives up control of her states one by one in that order, then $\sat$ loses the game for the first time after giving up $s$. The number obtained is then divided by the total number of such orderings.
We can also look at this definition from a probabilistic point of view: The importance of a state $s$ is the probability that, if $\sat$ gives up control of the states sequentially in an order drawn uniformly at random, the first time $\sat$ is no longer able to win the game is when she gives up control of $s$. This is what we call \emph{switching} in Definition~\ref{defImportancestate} below. 

\begin{definition}[Importance]
\label{defImportancestate}
The \emph{importance} for $\sat$ of a state $s \in \vsat$ with respect to a game $\game{}$  on an arena $((S,\Delta),\vsat,\vunsat)$ is defined as \[ \importance{s} = \frac{1}{n!} \sum_{\pi \in \Pi_{\vsat}} \val{S^\pi_{\geq s}} - \val{S^\pi_{\geq s} \setminus \set{s}} \]

where $n = \size{\vsat}$,  $\Pi_{\vsat}$ is the set of bijections from $\vsat$ to $\set{1,\ldots,n}$, and $S^\pi_{\geq s} = \set{s' \in \vsat \mid \pi(s') \geq \pi(s)}$. 

An equivalent definition, obtained by deleting the null terms from the sum, is obtained through the notion of \emph{critical pair}. A pair $(s,T) \in \vsat \times 2^{\vsat}$ is \emph{critical} if $\val{T \cup \set{s}} = 1$ and $\val{T} = 0$. Then we set 

\[\importance{s} = \frac{1}{n!} \sum_{(s,T)\text{ critical}} (\size{T})!(n-\size{T}-1)!\]

as $(\size{T})!(n-\size{T}-1)!$ is the number of $\pi \in \Pi_{\vsat}$ such that $S^\pi_{\geq s} \setminus \set{s} = T$.
\end{definition}

We say that $s$ \emph{switches the value} in $\pi$ if $(s, S^\pi_{\geq s} \setminus \set{s})$ is a critical pair. 
The importance of a state $s$ can then be seen as the proportion of orderings $\pi$ of the states in which $s$ switches the value. 

The following lemma states that if a state $s$ needs another state $s'$ (meaning that a set of states containing $s$ but not $s'$ always has value $0$), then the importance of $s'$ is at least as large as the one of $s$.

\begin{lemma}
Let $s,s'$ be two states of $\game{}$. If for all $T \subseteq \vsat$ such that $s \in T$ and $s' \notin T$, we have $\val{T}=0$ then $\importance{s} \leq \importance{s'}$. 
\end{lemma}

\begin{proof}
	Let $T \subseteq \vsat$ and suppose that $(s,T)$ is critical. Then $T$ necessarily contains $s'$ (as $\val{T \cup \set{s}} = 1$) and therefore $(s',T \cup \set{s} \setminus \set{s'})$ is critical. We can thus construct an injection associating to each $T$ such that $(s,T)$ is critical the set $T' =  T \cup \{s\} \setminus \{s'\}$ of equal size such that $(s', T')$ is critical. 
	
	We conclude using the second formula for the importance in \Cref{defImportancestate}.
\end{proof}

We now assume that we are given a partition $S_1, \ldots, S_n$ of $\vsat$. We generalize the previous definitions in a straightforward manner. We simply replace states with parts of the partition in the definitions, considering the $S_i$ as atomic elements. In all of our complexity proofs we will show the lower bounds for the previous case (in which states are partitioned in singletons) and the upper bounds for the general case. Thus all complexity results hold for both cases.

\begin{definition}[Importance for partitions]
	\label{defImportance}
	The \emph{importance} for $\sat$ of a set of states $S_i$ with $1 \leq i \leq n$ is defined as
  \[ \importance{S_i} = \frac{1}{n!} \sum_{\pi \in \Pi_n} \val{S^\pi_{\geq i}} - \val{S^\pi_{\geq i} \setminus S_i} \]
	where $\Pi_n$ stands for the set of permutations of $\set{1, \ldots,n}$, and
 \[S^\pi_{\geq i} = \bigcup_{\substack{1 \leq j \leq n\\ \pi(j) \geq \pi(i)}} S_j\]
	We define a pair $(i,J) \in \{1,\ldots, n\} \times 2^{\set{1, \ldots, n}}$ to be \emph{critical} if $\val{\bigcup_{j \in J \cup \set{i}} S_j} = 1$ and $\val{\bigcup_{j \in J} S_j} = 0$. Then we have:
	\[\importance{S_i} = \frac{1}{n!} \sum_{(i,J)\text{ critical}} \size{J}!(n-\size{J}-1)!\]
\end{definition}

Now let us show some basic results stating that parts with importance $0$ can be ignored in the computation of the importance of the other parts.

\begin{remark}
	\label{UselessStatesRemark}
	Let  $1 \leq i \leq n$. If $\importance{S_i} = 0$ then there is no $J \subseteq \set{1 \ldots, n}$ such that $(i,J)$ is critical. As a consequence, for all $J \subseteq \set{1, \ldots, n}$, $\val{\bigcup_{j \in J} S_j} = \val{\bigcup_{j \in J \cup \set{i}} S_j}$.
  This means that if $\importance{S_i} = 0$, then $\sat$ can always give up control of states $S_i$ without any effect on whether she wins the game.
\end{remark}

\begin{lemma}[Restriction to useful parts]
	\label{UselessStatesLemma}
	Let $I \subseteq \set{1 \ldots, n}$ be such that for all $j \notin I$, $\importance{S_j} = 0$. 
	Then  we have for all $i \in I$
	\[ \importance{S_i} = \frac{1}{\size{I}!} \sum_{\pi \in \Pi_I} \val{S^\pi_{\geq i}} - \val{S^\pi_{\geq i} \setminus S_i}, \]
	with $\Pi_I$ the set of bijections from $I$ to $\set{1,\ldots,\size{I}}$.
\end{lemma}

\begin{proof}
	For all $\pi \in \Pi_n$, let us denote by $\restrict{\pi}{I} : I \to \set{1, \ldots, \size{I}}$ the bijection such that for all $i,j \in I$, $\pi(i) < \pi(j)$ if and only if $\restrict{\pi}{I}(i) < \restrict{\pi}{I}(j)$.
	
	Note that for all $i \in I$ and $\pi \in \Pi$ we have 
	\[S^\pi_{\geq i} \setminus S^{\pi|_{I}}_{\geq i} \subseteq \bigcup_{j\in \set{1, \ldots,n} \setminus I} S_j.\] 
	As a consequence, $\val{S^\pi_{\geq i}} = \val{S^{\pi|_{I}}_{\geq i}$}, using Remark~\ref{UselessStatesRemark}. Similarly we get $\val{S^\pi_{\geq i} \setminus S_i} = \val{S^{\pi|_{I}}_{\geq i}\setminus S_i$}. This allows us to rewrite the importance of $S_i$ as
	\begin{align*}
		\importance{S_i} &= \frac{1}{n!} \sum_{\pi \in \Pi_n} \val{S^\pi_{\geq i}} - \val{S^\pi_{\geq i} \setminus S_i} \\
			&=\frac{1}{n!}  \sum_{\pi \in \Pi_n} \val{S^{\pi|_{I}}_{\geq i}} - \val{S^{\pi|_{I}}_{\geq i} \setminus S_i}  \\
			&= \frac{1}{n!} \sum_{\pi' \in \Pi_I} \frac{n!}{\size{I}!}  \cdot(\val{S^{\pi'}_{\geq i}} - \val{S^\pi_{\geq i} \setminus S_i})\\
			&= \frac{1}{\size{I}!} \sum_{\pi' \in \Pi_I} \val{S^{\pi'}_{\geq i}} - \val{S^{\pi'}_{\geq i} \setminus S_i}
	\end{align*}
	as for all $\pi' \in \Pi_I$ there are $\frac{n!}{\size{I}!}$ permutations $\pi \in \Pi_n$ such that $\restrict{\pi}{I} = \pi'$.
\end{proof}

\begin{corollary}
	\label{UselessStatesCorollary}
	Just as in Definition~\ref{defImportance}, by deleting the null terms from the sum we can rewrite the sum from Lemma~\ref{UselessStatesLemma}.
	Let $I \subseteq \set{1 \ldots, n}$ be such that for all $j \notin I$, $\importance{S_j} = 0$. Then we have
	\[\importance{S_i} = \frac{1}{\size{I}!} \sum_{(i,J)\text{ critical}, J\subseteq I} (\size{J})!(n-\size{J}-1)!\]
	for all $i \in I$.
\end{corollary}

We will also need the following lemma, stating that the importance of a part of a system remains unchanged when the specification is replaced with its complement.

\begin{lemma}[Complement objective]
	\label{ComplementGame}
	Let $\overline{\game{}}$ be the game with the same arena and initial state as $\game{}$ but the complement objective $\overline{\Omega} = S^{\omega} \setminus \Omega$. 
	Then for all $1 \leq i \leq n$, the importance of $S_i$ is the same for games $\game{}$ and $\overline{\game{}}$.
\end{lemma}

We now define the four computational problems which we will study throughout this paper. The three first are decision problems, the fourth is a counting one:\\

\textbf{Value problem}
\[
\begin{cases}
	\text{Input: }& \text{A game }\game{}\text{, a subset } \vsat'\subseteq\vsat \\
	\text{Output: }& \text{Do we have } \val{\vsat'} = 1?
\end{cases}
\]
\\

\textbf{Usefulness problem}
\[
\begin{cases}
	\text{Input: }& \text{A game }\game{} \text{, a partition } S_1, \ldots, S_n\\ 
	&\text{of the states, an index } i\\
	\text{Output: }& \text{Do we have } \importance{S_i} > 0?
\end{cases}
\]
\\

\textbf{Importance threshold problem}
\[
\begin{cases}
	\text{Input: }&\text{A game }\game{} \text{, a partition } S_1, \ldots, S_n\\
	&\text{of the states, an index } i, \eta \in \mathbb{Q}\\
	\text{Output: }& \text{Do we have } \importance{S_i} > \eta?
\end{cases}
\]
\\

\textbf{Importance computation problem}
\[
\begin{cases}
	\text{Input: }& \text{A game }\game{} \text{, a partition } S_1, \ldots, S_n\\
	&\text{of the states, an index } i\\
	\text{Output: }& n!\cdot\importance{S_i} 
\end{cases}
\]
\textit{The way the game is encoded is left open at this point, as it will depend on the specific kind of game in question, especially when it comes to the encoding of the objective.}

The two importance problems characterize the complexity of computing the importance of a state in a game. We will generally use the counting problem, except in cases where the complexity class obtained is more natural for the threshold version. For instance, if verifying some condition is already \expt-complete, then we want to say that the problem of computing how many elements of a set of exponential size respect that condition is also \expt-complete. However in order to do that we have to formulate the problem as a decision one.  
For the importance computation problem, the multiplication by $n!$ ensures that the output is always an integer, which is necessary in order for this to be a counting problem.

The usefulness problem is a restricted version of the importance threshold problem, only focusing on whether some part of the system may become necessary to the satisfaction of the specification when some other parts malfunction. A similar problem for voting games, called the pivot problem, has been studied in~\cite{KellyPrasad1990}.

\section{Importance values in LTL}

We now apply the theory developed in the preceeding section to linear time specifications in Kripke structures. It turns out that the three decision problems defined above are \twoexpt-complete for LTL specifications. As this renders practical applications essentially impossible, we then go on to investigate the problems when specifications are restricted to fragments of LTL, for which we obtain more tractable complexity classes.

\subsection{The full logic}

Let $\kripke = (S, \ap, \Delta, init, \lambda)$ be a Kripke structure and $\phi$ an LTL formula over $\ap$. 

\begin{definition}
\label{defltlgame}
Given a subset of states $\vsat \subseteq S$, let $\game{\vsat}$ be the game between players $\sat$ and $\unsat$ over the arena $((S,\Delta), \vsat, \vunsat)$ with $\vunsat = S \setminus \vsat$ and $init$ as initial state. The winning condition for player $\sat$ is the set of runs of $\kripke$ whose labeling satisfies $\phi$, i.e. $\set{r \in \runs{\kripke} \mid \lambda(r) \vDash \phi}$. The \emph{value} $\val{\vsat}$ of $\vsat \subseteq S$ is then defined as the value of $\vsat$ in the game $\game{\vsat}$ (see \Cref{defvalue}).
\end{definition}

Note that if one of the players owns all the states, then the game comes down to that player selecting a run in the structure. As a consequence, $\val{S} = 1$ if and only if $\kripke$ has a run satisfying $\phi$, and $\val{\emptyset} = 1$ if and only if all runs in $\kripke$ satisfy $\phi$.

\begin{definition}
\label{importanceltl}
Given a partition $S_1, \ldots, S_n$ of $S$, we define the \emph{importance} of a set of states $S_i$ with respect to LTL formula $\phi$ as the importance of $S_i$ in game $\game{S}$ under the same partition (see \Cref{defImportance}).
\end{definition}

A straightforward telescope sum argument shows that $\sum_{i=1}^n \importance{S_i} = val(S) - val(\emptyset)$. Therefore we have $\sum_{i=1}^n \importance{S_i} = 1$ if and only if there exists a run in $\kripke$ that satisfies $\phi$, but not all runs satisfy $\phi$. Otherwise the sum is $0$.

The intuition behind these definitions is that the value of a subset of states is $1$ if its elements can cooperate to guarantee the satisfaction of the specification no matter how the other states behave.
The importance of a state is high if it is critical in small subsets, or numerous subsets. We now illustrate our importance notion with a number of examples.

\begin{example}
\label{IntroExample}
	
	Let us first consider the examples given in the introduction and depicted in \Cref{fig:intro}, with states partitioned into singletons. Again we consider the specification $\phi = GF\,\mathbf{check} \land FG\,\neg \mathbf{fail}$, and we begin with the left-hand system involving only a single server $\mathbf{sv}$. Then $\sat$ wins the game $\game{\vsat}$ if and only if $ \{\mathbf{sv}, \mathbf{ok}\}\subseteq \vsat$: if $\sat$ is not in control of $\mathbf{sv}$, then $\unsat$ can respond $\mathbf{fail}$ forever, and if $\sat$ is not in control of $\mathbf{ok}$, then $\unsat$ can avoid further checks forever. Thus $(\mathbf{sv}, T)$ with $\mathbf{ok} \in T$ and $(\mathbf{ok}, T)$ with $\mathbf{sv}\in T$ are the only critical pairs, and it is straightforward to compute $\importance{\mathbf{ok}} = \importance{\mathbf{sv}} = 1/2$.
	
	Next consider the right-hand example of \Cref{fig:intro} involving two servers $\mathbf{sv}$ and $\mathbf{sv'}$. In this case $\sat$ wins the game $\game{\vsat}$ if and only if $\mathbf{ok}\in\vsat$ and $|\{\mathbf{sv},\mathbf{sv'}, \mathbf{check}\}\cap\vsat| \geq 2$. Namely, in this case $\mathbf{ok}$ can initiate infinitely many checks; if both servers can be controlled to respond correctly, then this automatically results in infinitely many successful checks, and if one server and $\mathbf{check}$ can be controlled, then $\mathbf{check}$ can choose the functioning server infinitely often. As $fail$ only has one outgoing transition, $\importance{\mathbf{fail}} = 0$, thus $\mathbf{fail}$ can be ignored by \Cref{UselessStatesLemma}. Each $s \in \{\mathbf{sv},\mathbf{sv'},\mathbf{check}\}$ yields two critical pairs $(s,T)$, where $|T| =2$, and so $\importance{s} = 1/6$. On the other hand, $\mathbf{ok}$ is the left part of every other critical pair and one then calculates $\importance{\mathbf{ok}} = 1/2$.
\end{example}

\begin{example}\label{ex:LTL}
In the three following examples we consider $\phi = a U b$, and the states are partitioned into singletons. 

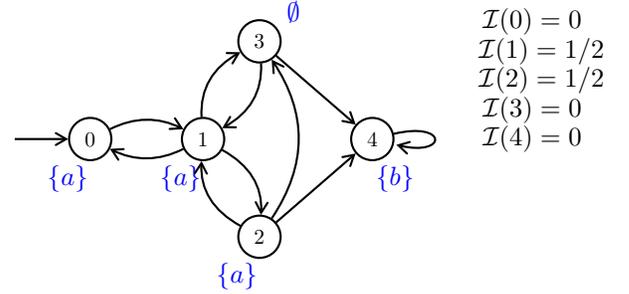
\begin{figure}[H]
    \centering
    \begin{tikzpicture}[xscale=1.5,yscale=1.3,AUT style]
\node[state,initial, scale=\myscale] (0) at (0,0) {$0$};
\node[state, scale=\myscale] (1) at (1,0) {$1$};
\node[state, scale=\myscale] (2) at (1.5,-1) {$2$};
\node[state, scale=\myscale] (3) at (1.5,1) {$3$};
\node[state, scale=\myscale] (4) at (2.5,0) {$4$};

\node (0') at (-0.2,-0.4) {\color{blue}$\set{a}$};
\node (1') at (0.8,-0.4) {\color{blue}$\set{a}$};
\node (2') at (1.3,-1.4) {\color{blue}$\set{a}$};
\node (3') at (1.8,1.3) {\color{blue}$\emptyset$};
\node (4') at (2.7,-0.4) {\color{blue}$\set{b}$};

\node (A) at (4, 1.2) {$\importance{0} = 0~~$};
\node (B) at (4, 0.9) {$\importance{1} = 1/2$};
\node (C) at (4, 0.6) {$\importance{2} = 1/2$};
\node (D) at (4, 0.3) {$\importance{3} = 0~~$};
\node (E) at (4, 0) {$\importance{4} = 0~~$};

\path[->, bend left] (0) edge node[right] {} (1);
\path[->, bend left] (1) edge node[right] {} (0);
\path[->, bend left] (2) edge node[right] {} (1);
\path[->, bend left] (1) edge node[right] {} (2);
\path[->, bend left] (3) edge node[above] {} (1);
\path[->, bend left] (1) edge node[left] {} (3);
\path[->, bend right] (2) edge node[right] {} (3);
\path[->] (2) edge node[above] {} (4);
\path[->] (3) edge node[right] {} (4);
\path[->, loop right] (4) edge node[right] {} (4);

\end{tikzpicture}
    \caption{Kripke structure of \Cref{ex:LTL} (1), where atomic propositions are displayed in blue, and importance values for $\phi = a U b$}
    \label{example2}
\end{figure}

(1) In the example of \Cref{example2} if $1$ and $2$ belong to $\sat$, then as every game starts with the transition from $0$ to $1$, she can then go from $1$ to $2$ and then to $4$, satisfying the specification.

However if $1$ belongs to $\unsat$, then $\unsat$ can win by indefinitely going back to $0$ from $1$. Similarly, if $2$ belongs to $\unsat$, then he can win by going from $2$ to $3$ if the game reaches $2$, leaving no possibility for $\sat$ to satisfy $a U b$.

As a result, a set of states will allow $\sat$ to win if and only if it contains $1$ and $2$, thus $1$ will be the one switching the value from $1$ to $0$ whenever it appears before $2$ in a permutation. This happens in half of the permutations, thus state $1$ has importance $1/2$ (see \Cref{defImportancestate} for what we mean by switching the value). Similarly, $2$ also has importance $1/2$. 

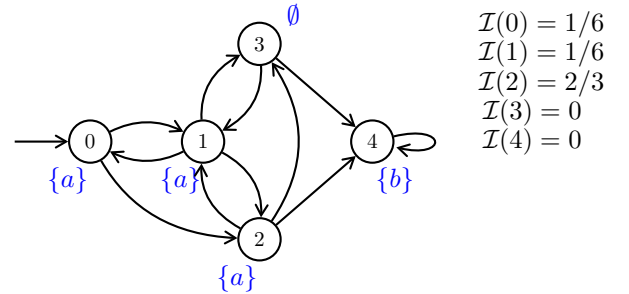
\begin{figure}[H]
    \centering
    \begin{tikzpicture}[xscale=1.5,yscale=1.3,AUT style]
\node[state,initial, scale=\myscale] (0) at (0,0) {$0$};
\node[state, scale=\myscale] (1) at (1,0) {$1$};
\node[state, scale=\myscale] (2) at (1.5,-1) {$2$};
\node[state, scale=\myscale] (3) at (1.5,1) {$3$};
\node[state, scale=\myscale] (4) at (2.5,0) {$4$};

\node (0') at (-0.2,-0.4) {\color{blue}$\set{a}$};
\node (1') at (0.8,-0.4) {\color{blue}$\set{a}$};
\node (2') at (1.3,-1.4) {\color{blue}$\set{a}$};
\node (3') at (1.8,1.3) {\color{blue}$\emptyset$};
\node (4') at (2.7,-0.4) {\color{blue}$\set{b}$};

\node (A) at (4, 1.2) {$\importance{0} = 1/6$};
\node (B) at (4, 0.9) {$\importance{1} = 1/6$};
\node (C) at (4, 0.6) {$\importance{2} = 2/3$};
\node (D) at (4, 0.3) {$\importance{3} = 0~~$};
\node (E) at (4, 0) {$\importance{4} = 0~~$};

\path[->, bend left] (0) edge node[right] {} (1);
\path[->, bend left] (1) edge node[right] {} (0);
\path[->, bend left] (2) edge node[right] {} (1);
\path[->, bend left] (1) edge node[right] {} (2);
\path[->, bend left] (3) edge node[above] {} (1);
\path[->, bend left] (1) edge node[left] {} (3);
\path[->, bend right] (2) edge node[right] {} (3);
\path[->, bend right] (0) edge node[right] {} (2);
\path[->] (2) edge node[above] {} (4);
\path[->] (3) edge node[right] {} (4);
\path[->, loop right] (4) edge node[right] {} (4);

\end{tikzpicture}
    \caption{Kripke structure of Ex. \ref{ex:LTL} (2) and importance values for $\phi = a U b$}
    \label{example3}
\end{figure}

(2) In the example of \Cref{example3} one can check that a set of states is allowing $\sat$ to win if and only if it contains $2$ and at least one of $0$ and $1$. Then $2$ will be the one switching the value in permutations where it appears before either $1$ or $0$, i.e. in $2/3$ of the permutations. In the other permutations the one switching the value is the second one to appear between $0$ and $1$.

\end{example}

We start our complexity results with the general case of an LTL specification. The complexities of the problems we consider is inferred from the \twoexpt-completeness of solving LTL games~\cite{Rosner1992}, which is inherited by the value problem. 

\begin{theorem}
\label{2exptimeltlimportance}
The usefulness and importance threshold problems for LTL with respect to Kripke structures are \twoexpt-complete. Further, one can compute the importance of a set of states in doubly exponential time.
\end{theorem}

\begin{proofsketch}
 The upper bound comes from the \twoexpt upper bound on solving LTL games and the fact that enumerating exponentially many permutations still stays within that class.
  The idea for the lower bound is to reduce the problem of solving an LTL game to the usefulness problem (with states partitioned into singletons). We consider an LTL game with states split between $\vsat$ and $\vunsat$. We add states $c_s, c_u$ and $t$ which are visited at the beginning of the game, and we add transitions from $c_s$ to states of $\vunsat$ and $c_u$ to states of $\vsat$. Finally, we add a sink state and a transition to it from every state. We partition states into singletons. See Figure~\ref{figureltl} for an illustration.
	
	We encode in the specification that one of the player wins automatically as soon as $c_s$ does not belong to $\sat$ or $c_u$ to $\unsat$.
	Let $T$ be a set of states of the game and assume that one of the states of $\vsat$ is not in $T$. Then we also encode in the specification that $\unsat$ can win by jumping from $c_u$ to that state and then to $sink$, making $\sat$ lose with both $T$ and $T\cup\set{t}$. Similarly we ensure that in order for $(t,T)$ to be critical, $T$ has to be disjoint from $\vunsat$. The only case in which $(t,T)$ can be critical is then the case where states are correctly distributed between the players, and the usefulness of $t$ is then equivalent to $\sat$ winning the original game. 
\end{proofsketch}
	
	\begin{figure}
		\centering
		\begin{tikzpicture}[xscale=1.5,yscale=1.3,AUT style]
\draw[thick] (0,0) ellipse (1.5cm and 1cm);
\draw[thick, dashed] (-1.5,0) -- (1.5,0);

\node[state,initial, scale=\myscale] (1) at (-3.6,0) {$c_s$};
\node[state, scale=\myscale] (2) at (-2.9,0) {$c_u$};
\node[state, scale=\myscale] (3) at (-2.2,0) {$t$};
\node[state, scale=\myscale] (4) at (-0.8,0.4) {$init$};

\node (A) at (0,0.5) {$V_{\sat}$};
\node (B) at (0,-0.5) {$V_{\unsat}$};
\node (C) at (1,1) {$\game{}$};

\node (A') at (-0.7,0.8) {};
\node (B') at (-0.7,-0.8) {};

\path[->] (1) edge node[right] {} (2);
\path[->] (2) edge node[right] {} (3);
\path[->, bend left=60] (3) edge node[right] {} (4);
\path[->,bend left=50] (2) edge node[right] {} (A');
\path[->,bend right] (1) edge node[right] {} (B');
\end{tikzpicture}
		\caption{Illustration for the proof of Theorem~\ref{2exptimeltlimportance}. Every state has a transition to a sink state which is not shown here.}
		\label{figureltl}
	\end{figure}
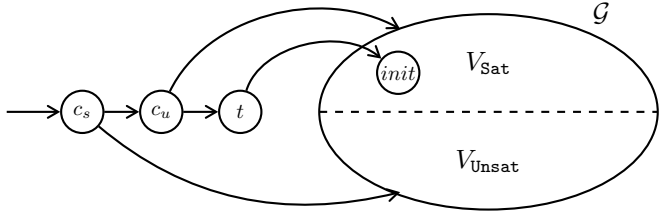

\subsection{Fragments of LTL}

Considering the high complexity of the computation of the importance in the case of LTL, we now look at fragments of the logic in order to get more tractable problems.
We therefore explore several classical winning conditions which can be expressed as LTL formulas. The value problem over Kripke structures with respect to some kind of specification is precisely the problem of deciding the winner of a game on a finite graph with such a specification as winning condition.

For the usefulness and importance problems, if the value problem has a complexity at least \pspace, we can enumerate permutations of states while keeping the same complexity. However, if the value problem is for instance in \poly or \np, then the complexity of the usefulness and importance problems is more involved.

Below, we study various types of winning conditions. We start with the basic case of reachability conditions, which allows us to also prove tight complexity bounds for Büchi, Muller and parity conditions. We consider here explicit Muller conditions, i.e., the condition is encoded as a list of sets of states. Muller conditions are sometimes encoded in more concise forms, such as a coloring function. We will give the complexity of that version as a consequence of the Emerson-Lei case, studied later in the paper.

\begin{proposition}
\label{valueReachability}
	The value problems for reachability, Büchi and explicit Muller conditions are \poly-complete.
\end{proposition}

\begin{proof}
	Reachability, Büchi and explicit Muller conditions are all known to be in $\poly$~\cite{Horn2008}. Furthermore, solving reachability games is known to be $\poly$-hard~\cite{Immerman1981}. 

	As we can encode the reachability condition reaching $f$ in all three winning conditions we consider here, we obtain \poly-hardness for those conditions.	
\end{proof}

\begin{remark}
	\label{valueParity}
	Solving games with parity conditions is in $\np \cap \conp$ \cite{emerson1993model}, but tight complexity bounds are not known, thus the same can be said about the value problem for parity conditions.
\end{remark}

\begin{proposition}
\label{usefulnessReach}
The usefulness problems for reachability, Büchi, parity and explicit Muller conditions with respect to Kripke structures are \np-complete.
\end{proposition}

\begin{proof}
	The problem is clearly in \np in the case of reachability, Büchi or Muller conditions as one can nondeterministically guess $J \subseteq \set{1,\ldots,n}$ and check in polynomial time whether $\val{\bigcup_{j \in J} S_j} = 0$ and $\val{\bigcup_{j \in J \cup \set{i}} S_j} = 1$ hold.
	
	For parity conditions we also have to guess positional strategies for $\sat$ and $\unsat$ along with $J$ and check in polynomial time that those strategies allow $\sat$ to win when she owns $\bigcup_{j \in J \cup \set{i}} S_j$ and $\unsat$ to win when $\sat$ owns $\bigcup_{j \in J} S_j$.
	
	We obtain \np-hardness through a reduction from 3SAT. Let $\psi = C_1 \land C_2 \land \cdots \land C_k$ be a 3SAT instance, with $C_j = (\ell_j^1 \lor \ell_j^2 \lor \ell_j^3)$ for all $j$, and let $\set{x_1, \ldots, x_n}$ be the set of variables appearing in $\psi$.
	
	We consider the Kripke structure $\kripke = (S, \ap, \Delta, c_1, \lambda)$ with states partitioned into singletons, and
	
	\begin{itemize}
		\item $S = \set{f, s, sink} \cup \set{c_i \mid 1 \leq i \leq k}$ \\$\cup \set{\ell_i^p \mid 1\leq i \leq k, 1\leq p \leq 3} \cup \set{x'_j, \neg x'_j \mid 1 \leq j \leq n}$
		
		\item $\ap = \set{f}$
		
		\item $\lambda(f) = \set{f}$ and $\lambda(q) = \emptyset$ for all $q \neq f$
		
	\end{itemize}
		
\begin{align*}
			\Delta = &\set{(c_i,\ell_i^p) \mid 1 \leq p \leq 3, 1 \leq i \leq k}\\
			&\cup \set{(\ell_i^p,c_{i+1}) \mid 1 \leq p \leq 3, 1 \leq i \leq k-1}\\
			&\cup \set{(\ell_k^p,s) \mid 1 \leq p \leq 3} \cup
			\set{(s,x'_1), (s, \neg x'_1)}\\
			&\cup \set{(x'_j, x'_{j+1}), (\neg x'_j, x'_{j+1})\mid 1 \leq j \leq n-1}\\
			&\cup \set{(x'_j, \neg x'_{j+1}), (\neg x'_j, \neg x'_{j+1})\mid 1 \leq j \leq n-1}\\
			&\cup \set{(x'_n,f), (\neg x'_n,f)} \cup \set{(q,sink) \mid q \in S} \\
			&\cup \set{(\ell_j^p, x'_m) \mid \ell_j^p \equiv \neg x_m}
\end{align*}

	Note that every literal in the clauses has a transition towards \textbf{its negation} in the variables. Player $\sat$ wins if and only if $f$ is reached, which can be expressed as a reachability, Büchi, parity or Muller condition. The construction can be done in logarithmic space. See \Cref{figureReach} for an illustration of the construction.
	
	We are now going to show that state $s$ is useful if and only if the 3SAT formula is satisfiable, thus proving \np-hardness of the usefulness problem for all four types of winning conditions.
	
	As every state has a transition to a sink state, if at some point a state belonging to $\unsat$ is reached before reaching $f$, then $\sat$ loses. As a consequence, $\sat$ wins with a set of states if and only if there is a path in this set of states from $c_1$ to $f$ (possibly not including $f$).
	
	Suppose there exists a valuation $\nu$ satisfying $\psi$. We extend $\nu$ to literals in the natural way, i.e. $\nu(\neg x_i) = \bot$ if $\nu(x_i) = \top$ and $\nu(\neg x_i) = \top$ otherwise. Then we set 
	\begin{align*}
		T = &\set{x'_m \mid \nu(x_m) = \top} \cup \set{\neg x'_m \mid \nu(x_m)=\bot}\\ \cup &\set{\ell_j^p \mid \nu(\ell_j^p)=\top} \cup \set{c_i \mid 1 \leq i \leq k}
	\end{align*} 

	Clearly there is a path from $c_1$ to $f$ in $T \cup \set{s}$, as for all $1 \leq i \leq k$ there is at least one $l_i^p$ satisfied by $\nu$ (and thus in~$T$), and for all $1 \leq i \leq n$ one of $x'_i, \neg x'_i$ is in $T$. However, for all $(\ell_i^p, x'_j) \in \Delta$ (resp. $(\ell_i^p,\neg x'_j)$), if $\ell_i^p \in T$ then $\nu(\ell_i^p)= \top$ thus, as $\ell_i^p \equiv \neg x_j$ (resp. $x_j$), $\nu(x_j)= \bot$ (resp. $\top$) and $x'_j \notin T$ (resp. $\neg x'_j$). Therefore there is no path in $T$ from $c_1$ to $f$.

	Now suppose there exists $T$ such that there is a path from $c_1$ to $f$ in $T\cup \set{s}$ but not in $T$. Then there is a path in $T\cup \set{s}$ from $c_1$ to $f$ going through $s$. In particular for all $1 \leq i \leq n$ at least one of $x'_i, \neg x'_i$ is in $T$. Let $\nu$ be a valuation such that for all $i$, if $\nu(x_i) = \top$ then $x'_i \in T$ and $\neg x'_i \in T$ otherwise. There is also a path from $c_1$ to $s$ in $T$, hence for all $i$ there is a $p_i$ such that $\ell_i^{p_i} \in T$. Then for all $\ell_i^{p_i}$ of the form $x_j$ for some $j$, the state $\neg x'_j$ cannot be in $T$ as otherwise there would be a path from $c_1$ to $\ell_i^{p_i}$ then to $\neg x'_j$ and finally to $f$ in $T$, not going through $s$. As a result we have $x'_j \in T$ and thus $\nu(l_i^{p_i}) = \nu(x_j) = \top$. By a similar argument, if $\ell_i^{p_i} = \neg x_j$ then $\nu(x_j) = \bot$. Hence for every $i$ there is a literal in the $i$th clause satisfied by $\nu$, thus the 3SAT instance is satisfiable.
	
	\begin{figure}
		\centering
		\begin{tikzpicture}[xscale=1.5,yscale=1.3,AUT style]

\node[state,initial, scale=\myscale] (c1) at (0,0) {$c_1$};
\node[state, scale=\myscale] (x1) at (1,1) {$x_1$};
\node[state, scale=\myscale] (nx2) at (1,-1) {$\neg x_2$};
\node[state, scale=\myscale] (nx1) at (1,0) {$\neg x_1$};
\node[state, scale=\myscale] (s) at (2,0) {$s$};
\node[state, scale=\myscale] (x1') at (2.5,-1) {$x_1'$};
\node[state, scale=\myscale] (nx1') at (3.5,-1) {$\neg x_1'$};
\node[state, scale=\myscale] (x2') at (2.5,-1.7) {$x_2'$};
\node[state, scale=\myscale] (nx2') at (3.5,-1.7) {$\neg x_2'$};
\node[state, scale=\myscale] (f) at (3,-2.2) {$f$};

\path[->] (c1) edge node[right] {} (x1);
\path[->] (c1) edge node[right] {} (nx1);
\path[->] (c1) edge node[right] {} (nx2);
\path[->] (x1) edge node[right] {} (s);
\path[->] (nx1) edge node[above] {} (s);
\path[->] (nx2) edge node[above] {} (s);
\path[->, bend left] (x1) edge node[right] {} (nx1');
\path[->] (nx1) edge node[above] {} (x1');
\path[->] (nx2) edge node[above] {} (x2');
\path[->] (s) edge node[left] {} (x1');
\path[->] (s) edge node[left] {} (nx1');
\path[->] (nx1') edge node[left] {} (x2');
\path[->] (x1') edge node[left] {} (nx2');
\path[->] (x1') edge node[left] {} (x2');
\path[->] (nx1') edge node[left] {} (nx2');
\path[->] (x2') edge node[right] {} (f);
\path[->] (nx2') edge node[right] {} (f);
\end{tikzpicture}
		\caption{Construction for $(x_1 \lor \neg x_1 \lor \neg x_2)$. All states have a transition to a sink state, not shown here.}
		\label{figureReach}
	\end{figure}
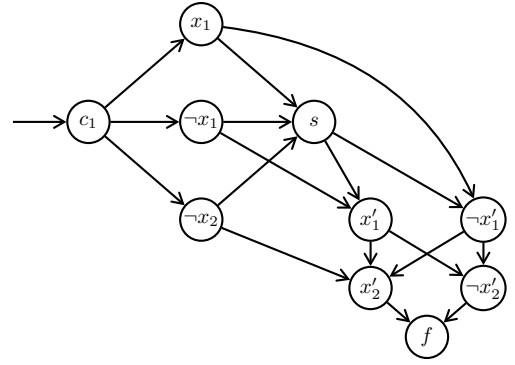  
\end{proof}


\begin{theorem}
\label{sharpPcompleteReachability}
The importance computation problems for reachability, Büchi, parity and explicit Muller conditions with respect to Kripke structures are \sharppoly-complete.
\end{theorem}

\begin{proofsketch}
	The idea is to reduce the problem of counting the valuations satisfying exactly one literal of every clause of a 3SAT formula $\phi$, known to be \sharppoly-complete \cite{Valiant1979}. First we transform the formula $\phi$ into another one $\psi$ that is satisfied by a valuation $\nu$ if and only if $\nu$ satisfies one literal per clause in $\phi$. We then reuse the construction of the usefulness proof, and notice that the set of sets of states $T$ making $(s,T)$ critical in the structure can be split into parts of (up to some details) equal size, each one matching a valuation satisfying the formula. Further, all such sets are of (again, up to some details) the same size. This allows us to compute the number of valuations satisfying the formula from the importance of $s$.
\end{proofsketch}

\begin{remark}
	One can show with nearly identical proofs that those problems keep the same complexity with co-Büchi, safety or co-safety conditions.
\end{remark}

Now we consider not only Büchi conditions, but Boolean combinations of them, called Emerson-Lei conditions. As expected, we get an intermediate complexity between those for Büchi and LTL conditions.

\begin{theorem}
	\label{EmersonLei}
	The value, usefulness and importance threshold problems for Emerson-Lei conditions are \pspace-complete. Further, one can compute the importance of a set of states in polynomial space.
\end{theorem}

\begin{proofsketch}
	As Emerson-Lei games are known to be in \pspace~\cite{HunterDawar2005}, the upper bound follows easily. We prove the lower bound by reduction of QSAT. We construct a structure encoding a sequence of choices of the values of the variables and partition its states into singletons. The structure contains a state $s$, which we will prove to be critical if and only if the QSAT formula is valid. We ensure that for all $T$, $(s,T)$ can only be critical if $T$ contains the states choosing the values of the existential variables and not the other ones, by making one of the players win without using $s$ for sets $T$ not satisfying this condition.
	
	We also ensure that $\unsat$ wins if he owns $s$, thus $s$ is useful if and only if $\sat$ wins with $s$. We make players choose valuations of the variables infinitely many times, and we encode in the specification that the player owning the first variable $x_i$ such that $x_i$ and $\neg x_i$ are chosen infinitely often loses. If both players play consistently, the game is decided by the satisfaction of the QSAT formula. 
\end{proofsketch}

\begin{remark}
	It was proven by Hunter and Dawar that Emerson-Lei conditions are more succinct than Muller conditions encoded with a coloring of the states and a list of sets of colors~\cite{HunterDawar2005}. As a result, the \pspace lower bounds we obtained for Emerson-Lei transfer to these succinct Muller conditions. Hunter and Dawar also show that solving games with those Muller conditions is \pspace-complete, from which we can easily infer the \pspace-completeness of the value, usefulness and importance threshold problem for this type of condition.
\end{remark}

We continue our exploration with a more complicated case, the Rabin and Streett conditions. We treat both cases simultaneously as they are symmetric. 

\begin{remark}
	As solving Rabin (resp. Streett) games is \np-complete (resp. \conp-complete), so is the value problem for Rabin (resp. Streett) conditions~\cite{emerson1988complexity}. 
\end{remark}

\begin{proposition}
	\label{usefulnessRabin}
	The usefulness problem for Rabin conditions is \sigmatwo-complete.
\end{proposition}

\begin{proofsketch}
	The complete proof is in the appendix. We reduce the dual of the $\forall \exists$3SAT problem to the usefulness problem in the case when states are partitioned in singletons. The set of states $T$ witnessing the usefulness of state $s$ will encode the valuation of the first set of variables, with a trick similar to the one used in the proof of Proposition \ref{usefulnessReach} to ensure that the encoded valuation is correct. 
	
	As $\sat$ plays for a Rabin objective, she has a positional strategy, with which she has to choose for each clause a satisfied literal. We use the Rabin condition to make sure that $\sat$ does not pick a literal and its negation. We also ensure that $\sat$ wins automatically with $T \cup \set{s}$ as soon as $T$ encodes a correct valuation, and then $s$ is useful if and only if there exists a set of states $T$ (i.e. a valuation of the first variables) such that for all positional strategy of $\sat$ over $T$ (i.e. valuation of the second variables), $\sat$ loses the game (i.e. the formula is not satisfied).
\end{proofsketch}

The theorem below uses the complexity class \sharpsigmatwo, which is the class of counting problems $P$ such that there exists a nondeterministic polynomial-time Turing machine with an \np oracle such that the answer of $P$ on an input is the number of accepting runs of the machine on that input.

\begin{theorem}
	\label{importanceRabin}
	The importance computation problem for Rabin conditions is \sharpsigmatwo-complete.
\end{theorem}

\begin{proofsketch}
	The idea is simply to observe that in the construction for Proposition~\ref{usefulnessRabin}, the sets of states witnessing the usefulness of $s$ are in bijection with the valuations of the universal variables witnessing the non-validity of the $\forall \exists 3Sat$ formula (up to some technical details). In the appendix we show that counting such valuations is \sharpsigmatwo-complete, from which one can infer \sharpsigmatwo-completeness of the importance computation problem.
\end{proofsketch}

\begin{corollary}
	\label{usefulnessImportanceStreett}
	As Streett conditions are exactly the complements of Rabin ones, by Lemma~\ref{ComplementGame} and Proposition~\ref{usefulnessRabin},
	the usefulness problem for Streett conditions is \sigmatwo-complete.
	
	By the same argument, by Lemma~\ref{ComplementGame} and Theorem~\ref{importanceRabin}, the importance computation problem for Streett conditions is \sharpsigmatwo-complete.
\end{corollary}

\section{Importance values in CTL}

We now adapt the definitions to deal with \ctl specifications. A notion of degree of responsibility of a state in a Kripke structure for the satisfaction of a \ctl formula was already given by Chockler, Halpern, and Kupferman~\cite{Causality2008}. While in their approach the responsibility of a state was based on the set of atomic propositions it chooses to satisfy, in ours it is based on the set of outgoing transitions it chooses to allow.

In contrast to the previous sections, CTL has the additional challenge that the formulas are evaluated on trees and not on words.
The first question that arises is the nature of the nondeterministic choices in this setting.
Our definitions rely on the fact that the nondeterminism of a state may be resolved in different ways by the two players.
However, due to the branching time nature of CTL, directly applying this methodology does not make sense, as CTL formulas already take the nondeterminism into account. 
This is why we consider \emph{modal transition systems} (MTS, as introduced in \Cref{StructDefinitions}), in which there is another layer of choice: namely determining the subset of \emph{may} transitions that are present in any state.
Modal transition systems have been widely studied as a formalism to capture the refinement of processes from abstract specifications to concrete implementations~\cite{Kretinsky2017,LarsenT1988}.
They have been extended in various ways, and the corresponding synthesis and verification problems have been considered~\cite{BenesCK2011,BauerFJLLT2011,AntonikHLNW2008}.

The second, and related, issue is that letting the players construct the tree turn-by-turn runs into the problem that the order in which different branches are considered will often make a difference.
In Section~\ref{section:witnessgame} we explain the difficulties of defining a game which allows both players to construct a tree generated by an MTS in a similar fashion as for LTL. 

In Section~\ref{section:2turngame} we define a notion of importance for CTL on MTS, which we call \emph{two-turn CTL}, where both players choose once in the beginning which $may$-transitions they allow in the states under their control.
This choice induces a Kripke structure on which the CTL formula can be evaluated.

However the order in which the choices are made affects the importance values. Therefore, in Section~\ref{section:concurrent} we consider also the concurrent setting in which randomized strategies become important, which we call \emph{Concurrent CTL}. 

Throughout, let $\modal = (S, \ap, \Delta_{must}, \Delta_{may}, init, \lambda)$ be a modal transition system and let $\phi$ be a \ctl formula.

\subsection{Importance in an MTS with respect to a CTL specification}
\label{section:witnessgame}

It is appealing to define a notion of importance that relies on the ability of a set of states to \textit{guarantee} the satisfaction of a specification. However, in the case of \ctl the fact that formulas are evaluated on trees and not on runs forbids us to make the players construct a tree in a similar fashion as for $\ltl$, as the winning player would heavily depend on the order in which branches are constructed. This is illustrated in the following example.

\begin{example}\label{ex:CTL problem}
	Consider the formula $[E F 3 \land EF 4 \land AG (\neg (1 \land EX 3 \land EX 4))] \lor EF (2 \land AX 5)$ and the MTS in \Cref{figureWitness}. Say we want to make the players construct a tree, and $\sat$ wins if this tree satisfies the above formula.
  We would have to decide in which order the players make their choices.
	Say we make $\sat$ choose the successors of $1$ before $\unsat$ chooses the ones of $2$. Then $\unsat$ wins by picking at least one of $3$ and $4$ as successor if $1$ has neither of them as successor, and the same set of successors as $1$ otherwise. If $\unsat$ chooses first, then $\sat$ can pick $3$ and not $4$ if $\unsat$ chose $4$, and vice-versa (and pick any one of the two if $\unsat$ chose both or none).
	Thus the winner depends on the (arbitrary) order in which we let the players construct the tree. 
	
	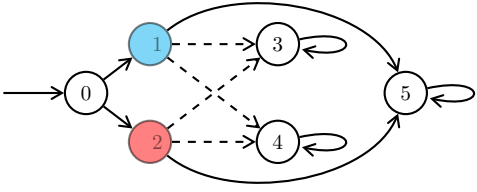
\begin{figure}
		\centering
		\begin{tikzpicture}[xscale=1.7,yscale=1.3,AUT style]
\node[state,initial, scale=\myscale] (0) at (-0.5,-0.5) {$0$};
\node[state, scale=\myscale, fill=cyan,opacity=.5,text opacity=1] (1) at (0,0) {$1$};
\node[state, scale=\myscale, fill = red,opacity=.5,text opacity=1] (2) at (0,-1) {$2$};
\node[state, scale=\myscale] (3) at (1,0) {$3$};
\node[state, scale=\myscale] (4) at (1,-1) {$4$};
\node[state, scale=\myscale] (5) at (2,-0.5) {$5$};

\path[->] (0) edge node[right] {} (1);
\path[->] (0) edge node[right] {} (2);
\path[->, dashed] (1) edge node[right] {} (3);
\path[->, dashed] (1) edge node[right] {} (4);
\path[->, dashed] (2) edge node[right] {} (3);
\path[->, dashed] (2) edge node[right] {} (4);
\path[->, loop right] (3) edge node[right] {} (3);
\path[->, loop right] (4) edge node[right] {} (4);
\path[->, bend left= 60] (1) edge node[right] {} (5);
\path[->, bend right=60] (2) edge node[right] {} (5);
\path[->, loop right] (5) edge node[right] {} (5);
\end{tikzpicture}
		\caption{A modal transition system with $must$-transitions depicted as solid lines and $may$-transitions depicted as dashed lines. This example illustrates the problem that has to be faced when defining turn-based CTL values (cf. \Cref{ex:CTL problem}). The state in blue belongs to $\sat$, the state in red to $\unsat$.}
		\label{figureWitness}
	\end{figure}
\end{example}

We could define a success value where some sets of states are seen as neutral, meaning that when $\sat$ has this group of states the game is undetermined, similarly to what was done in~\cite{HuthJS2001}. However we wish to define the importance as a numerical value, thus it is more practical that the success value can only be $0$ or $1$.

\subsection{Two-turn CTL importance}
\label{section:2turngame}

The idea of two-turn importance values is that a set of states has value one if it can choose sets of outgoing transitions such that the specification is satisfied no matter which sets of outgoing transitions are chosen by the other states. This definition puts more burden on the satisfier, but matches a vision of the MTS as a way to represent a set of Kripke structures (possible implementations of a system) rather than a language of trees. 

\begin{definition}[Two-turn importance values]
	\label{2turnAsymetricCTL}
	Let \mbox{$\modal = (S,\ap, \Delta_{must}, \Delta_{may}, init, \lambda)$} be a modal transition system, let $\vsat \subseteq S$ and let $\phi$ be a \ctl specification. A \emph{pure strategy} for $\sat$ is a function $\sigma_{\sat} : \vsat \to 2^{\Delta_{may}}$ such that for all $v \in \vsat$, $\Delta_{must}(v) \subseteq \sigma_{\sat}(v) \subseteq \Delta_{may}(v)$. We define pure strategies $\sigma_{\unsat}$ for $\unsat$ symmetrically. 
	
	Two pure strategies $\sigma_{\sat}, \sigma_{\unsat}$ yield a Kripke structure, whose states are the ones of $\modal$ and transitions from a state are given by the strategy of the player owning that state. We call that Kripke structure $\kripke(\sigma_{\sat}, \sigma_{\unsat})$.
	
	The value $\valtwoturn{\vsat}$ of $\vsat$ is defined as $1$ if there exists a pure strategy $\sigma_{\sat}$ of $\sat$ such that for all pure strategies $\sigma_{\unsat}$ of $\unsat$, $\kripke(\sigma_{\sat},\sigma_{\unsat})$ satisfies $\phi$, and $0$ otherwise. 
	
	The \emph{importance} is defined analogously to \Cref{defImportance}: Given a partition $S_1,\ldots, S_n$ of $S$, the importance of $S_i$ is defined as:
	\[ \importancetwoturn{S_i} = \frac{1}{n!} \sum_{\pi \in \Pi_n} \valtwoturn{S^\pi_{\geq i}} - \valtwoturn{S^\pi_{\geq i} \setminus S_i} \]
\end{definition}

\begin{example}
	\label{CTLexamples}
	
	(1) Consider the formula $\phi = A (EF a) U b$ and the modal transition system displayed in \Cref{example5}. Observe that the two ways in which $\phi$ may be violated are
	\begin{itemize}
		\item $\unsat$ owns $0$, $2$ and $5$ and allows the transition from $0$ to $5$ but not from $5$ to $4$ or $2$ to $4$, so that there is a path labeled $\set{a}\emptyset$ to $5$, but no path from $5$ to a state labeled by an $a$.
		
		\item $\unsat$ owns $1$ and $2$ and chooses transitions so that there is no transition from $1$ to $3$ or from $2$ to $4$.
		
	\end{itemize}  The importance is therefore distributed as follows:
	
	\begin{figure}[H]
		\centering
		\begin{tikzpicture}[xscale=1.7,yscale=1.3,AUT style]
\node[state,initial, scale=\myscale] (0) at (-1.2,0) {$0$};
\node[state, scale=\myscale] (1) at (1,0) {$1$};
\node[state, scale=\myscale] (2) at (1.5,-1) {$2$};
\node[state, scale=\myscale] (3) at (2,0) {$3$};
\node[state, scale=\myscale] (4) at (0.5,-1) {$4$};
\node[state, scale=\myscale] (5) at (0.5,-2.5) {$5$};
\node[state, scale=\myscale] (6) at (-0.3, -1) {$6$};

\node (0') at (-1.4,-0.4) {\color{blue}$\set{a}$};
\node (1') at (0.8,-0.4) {\color{blue}$\emptyset$};
\node (2') at (1.7,-1.4) {\color{blue}$\set{b}$};
\node (3') at (2.2,-0.4) {\color{blue}$\set{a}$};
\node (4') at (0.3,-1.4) {\color{blue}$\emptyset$};
\node (5') at (0.3,-2.9) {\color{blue}$\emptyset$};
\node (6') at (-0.4,-1.5) {\color{blue}$\set{a,b}$};

\node (A) at (-2,-1.2) {$\importancetwoturn{0} = 1/12$};
\node (B) at (-2,-1.5) {$\importancetwoturn{1} = 1/4~~$};
\node (C) at (-2,-1.8) {$\importancetwoturn{2} = 7/12$};
\node (D) at (-2,-2.1) {$\importancetwoturn{3} = 0~~~~$};
\node (E) at (-2,-2.4) {$\importancetwoturn{4} = 0~~~~$};
\node (F) at (-2,-2.7) {$\importancetwoturn{5} = 1/12$};
\node (G) at (-2,-3) {$\importancetwoturn{6} = 0~~~~$};

\path[->, bend left] (0) edge node[right] {} (1);
\path[->, bend left] (0) edge node[right] {} (4);
\path[->, bend right, dashed] (0) edge node[right] {} (5);
\path[->, bend left] (1) edge node[right] {} (2);
\path[->, bend left, dashed] (1) edge node[right] {} (3);
\path[->, bend left,dashed] (2) edge node[right] {} (4);
\path[->, loop right] (2) edge node[right] {} (2);
\path[->, bend left] (3) edge node[above] {} (2);
\path[->, bend left] (4) edge node[left] {} (2);
\path[->] (4) edge node[above] {} (6);
\path[->,dashed] (5) edge node[right] {} (4);
\path[->] (5) edge node[right] {} (2);
\path[->, loop left] (6) edge node[right] {} (6);

\end{tikzpicture}
		\caption{MTS of Ex. \ref{CTLexamples} (1) and 2-turn importance values for $\phi = A (EF a) U b$}
		\label{example5}
	\end{figure}
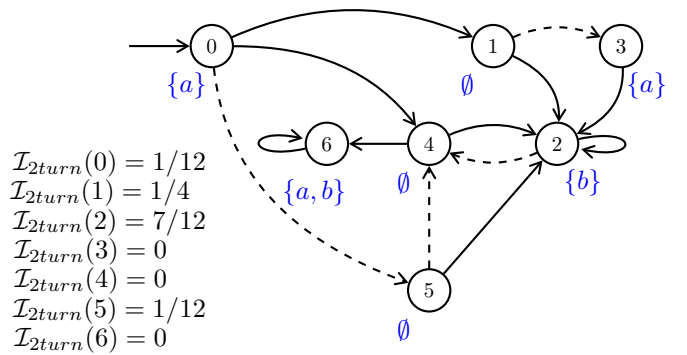
	
	(2) In the example of \Cref{example6} we want to illustrate a limitation of this notion with respect to what was discussed in Section~\ref{section:witnessgame}.
	Such a mechanism can be illustrated by trying to prove $AG (a\Rightarrow EX(EF b))$ on the following structure:
	
	\begin{figure}[H]
		\centering
		\begin{tikzpicture}[xscale=1.7,yscale=1.3,AUT style]
\node[state,initial, scale=\myscale] (0) at (-0.4,0) {$0$};
\node[state, scale=\myscale] (1) at (0.5,0) {$1$};
\node[state, scale=\myscale] (2) at (1.5,0) {$2$};
\node[state, scale=\myscale] (3) at (-0.4,-1) {$3$};
\node[state, scale=\myscale] (4) at (0.5,-1) {$4$};

\node (0') at (-0.1,0.3) {\color{blue}$\emptyset$};
\node (1') at (0.7,-0.3) {\color{blue}$\emptyset$};
\node (2') at (1.7,-0.4) {\color{blue}$\set{b}$};
\node (3') at (-0.3,-1.5) {\color{blue}$\set{a,b}$};
\node (4') at (0.8,-1.4) {\color{blue}$\set{a}$};

\node (A) at (3,0) {$\importancetwoturn{0} = 1/3$};
\node (B) at (3,-0.3) {$\importancetwoturn{1} = 0~~~$};
\node (C) at (3,-0.6) {$\importancetwoturn{2} = 0~~~$};
\node (D) at (3,-0.9) {$\importancetwoturn{3} = 1/3$};
\node (E) at (3,-1.2) {$\importancetwoturn{4} = 1/3$};

\path[->] (0) edge node[right] {} (1);
\path[->, bend left, dashed] (0) edge node[right] {} (3);
\path[->, dashed] (1) edge node[right] {} (2);
\path[->, loop above] (1) edge node[right] {} (1);
\path[->, loop right] (2) edge node[right] {} (2);
\path[->, bend left] (3) edge node[above] {} (0);
\path[->,dashed] (3) edge node[above] {} (4);
\path[->, loop below] (4) edge node[right] {} (4);
\path[->, bend right,dashed] (4) edge node[left] {} (2);

\end{tikzpicture}
		\caption{MTS of Ex. \ref{CTLexamples} (2) and 2-turn importance values for \hspace{\textwidth}  $\phi = AG (a\Rightarrow EX(EF b))$}
		\label{example6}
	\end{figure}
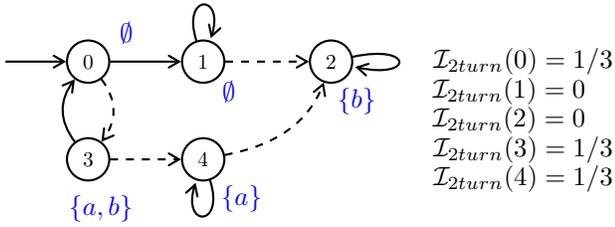
	
	In the 2-turn CTL framework, $\unsat$ wins if and only if he owns $0$, $3$ and $4$ in order to create a path to $4$, but no transition from $4$ to $2$. Thus in any ordering of the states the last one between $0$, $3$ and $4$ will be the one switching the value.
	
	However, one might want to design a richer model in which we would also give the victory to $\unsat$ when he owns either $0$, $1$ and $3$ or $0$, $1$ and $4$. The reason is that after allowing the transition from $0$ to $3$ at the start, we would like to let $\unsat$ delete it. Then $\unsat$ can not allow the transition from $1$ to $2$, and ensure that there is no path from $3$ to $2$. Therefore there is no path from a successor of $3$ reaching a state labeled $b$. 
	
	This observation motivates the study of turn-based definitions of CTL importance in MTS for restricted sets of formulas, which we leave open for future work.
\end{example}

In the appendix we prove the following results. The hardness proofs consist in encoding choices of valuations of variables in SAT formulas as the players' choices of transitions.  

\begin{proposition}
	\label{2turnCTLvalue}
	The value problem for two-turn CTL is \mbox{\sigmatwo-complete}.
\end{proposition}

\begin{proposition}
	\label{2turnCTLusefulness}
	The usefulness problem for two-turn CTL is \sigmathree-complete.
\end{proposition}

\begin{theorem}
	\label{2turnCTLimportance}
	The importance computation problem for two-turn CTL is \sharpsigmathree-complete.
\end{theorem}

\begin{remark}
	\label{dualCTLgame}
	We can define a dual game, in which $\unsat$ plays first, and then $\sat$. While in the former game $\sat$ was at a disadvantage, in this version $\unsat$ is, as he is the one who has to choose his strategy without knowing the adversary's. 
	
	 Let $\modal$ be an MTS with a set of states $S$, let $\vsat \subseteq S$ and let $\phi$ be a specification, the value of $\vsat$ with respect to $\phi$ in the game where $\unsat$ starts is $1 - \valtwoturn{S \setminus \vsat}$ with $\valtwoturn{S \setminus \vsat}$ the value of $S \setminus \vsat$ with respect to $\neg \phi$ in the game where $\sat$ starts.
	From this one infers easily that the value problem for the game where $\unsat$ starts is \pitwo-complete and that, by an argument similar to the proof of \ref{ComplementGame}, the usefulness problem is \sigmathree-complete and the importance computation problem \sharpsigmathree-complete.
\end{remark}

\subsection{Concurrent CTL importance}
\label{section:concurrent}

The previous version of the game breaks the symmetry between the two players: We have to pick either $\sat{}$ or $\unsat{}$ to play first (we chose $\sat{}$ in the definition above). One may prefer a version of this game in which we do not give any such advantage to a player.

We now introduce a concurrent game, in which both players choose a mixed strategy, in the form of a distribution over all the possible choices of sets of transitions from their respective states. The value of a set of states is the highest probability such a mixed strategy can guarantee for $\sat{}$ with this set of states. The Nash Theorem guarantees the existence of a Nash equilibrium, which means that the highest probability a set of states can achieve for $\sat{}$ is one minus the highest probability its complement can achieve for $\unsat{}$. For an introductory account on non-cooperative concurrent games, we refer to~\cite{Owen2013}.

\begin{definition}[Concurrent game induced by CTL formula]
Let $\modal = (S,\ap, \Delta_{must}, \Delta_{may}, init, \lambda)$ be a modal transition system, let $\vsat \subseteq S$ and let $\phi$ be a \ctl specification. Let $\choicessat{\modal}{\vsat}$ be the set of pure strategies for $\sat$. A \emph{mixed strategy} for $\sat{}$ is a probability distribution $p_{S} : \choicessat{\modal}{\vsat} \to [0,1]$. We define $\choicesunsat{\modal}{\vunsat}$ and $p_{U}$ in a similar way. Let $M_S$ and $M_U$ denote the set of mixed strategies of $\sat$ and $\unsat$, respectively.

We consider the concurrent game with the payoff functions
\[\rho_{\sat}(\sigma_{\sat}, \sigma_{\unsat}) = \begin{cases}
1 \text{ if } \kripke(\sigma_{\sat}, \sigma_{\unsat}) \text{ satisfies } \phi\\
0 \text{ otherwise}
\end{cases}\]
and $\rho_{\unsat}(\sigma_{\sat}, \sigma_{\unsat}) = 1 - \rho_{\sat}(\sigma_{\sat}, \sigma_{\unsat})$, for all $\sigma_{\sat} \in \choicessat{\modal}{\vsat}, \sigma_{\unsat} \in \choicesunsat{\modal}{\vunsat}$. Given two mixed strategies $p_{S}, p_{U}$, the  expected payoff of $\sat$ is
\begin{align*}
&E_{\sat}(\modal, \phi, p_{S}, p_{U}) = \\
&\sum_{\substack{\sigma \in \choicessat{\modal}{\vsat} \\ \sigma' \in \choicesunsat{\modal}{\vunsat}}} p_{S}(\sigma) p_{U}(\sigma')\rho_{\sat}(\sigma, \sigma')
\end{align*}
The expected payoff of $\unsat{}$ is \[E_{\unsat}(\modal, \phi, p_{S}, p_{U}) = 1- E_{\sat}(\modal, \phi, p_{S}, p_{U})\]. Finally, we define the value of a set of states $\vsat$ as \[\valconcurrent{\vsat} = \sup_{p_S\in M_S} \inf_{p_U \in M_U} E_{\sat}(\modal, \phi, p_{S}, p_{U})\]
\end{definition}

It is a direct consequence of Nash's Theorem~\cite{Nash1951} that $\valconcurrent{\vsat}$ is the payoff of $\sat$ obtained in any Nash equilibrium of the concurrent game defined above. In particular we have
	\[\valconcurrent{\vsat} = \inf_{p_U \in M_U} \sup_{p_S\in M_S} E_{\sat}(\modal, \phi, p_{S}, p_{U})\]

\begin{definition}[Concurrent importance values]
	\label{NashImportance}
	Given a partition of the states $S_1, \ldots, S_n$, we define the importance of a set of states $S_i$ as usual:
	\[ \importanceconcurrent{S_i} = \frac{1}{n!} \sum_{\pi \in \Pi_n} \valconcurrent{S^\pi_{\geq i}} - \valconcurrent{S^\pi_{\geq i} \setminus S_i} \]
\end{definition}

\begin{lemma}
	\label{NashComplexity}
	The importance value of a set of states of an MTS for a CTL specification with respect to \Cref{NashImportance} can be computed in exponential time. 
\end{lemma}

\begin{proof}
	We can check whether a Kripke structure satisfies a CTL specification in polynomial time. Hence we can compute the winner for all (exponentially many) pairs of pure strategies. 
	
	Computing the value of a set of states then amounts to solving a linear optimization problem with exponential input~\cite{ChenDeng2006}. As the latter problem can be solved in polynomial time, the former is in \expt~\cite{Khachiyan1979}.
\end{proof}

\begin{lemma}
	\label{Concurrentdeterminestwoturn}
	For each set of states $\vsat$, we have $\valconcurrent{\vsat} = 1$ if and only if $\valtwoturn{\vsat} =1$. 
	In particular, as $\valtwoturn{\vsat} \in \set{0,1}$ for all $\vsat$, the value $\valtwoturn{\vsat}$ is entirely determined by $\valconcurrent{\vsat}$ (it is its integer part).
\end{lemma}

\begin{proof}
	Suppose $\valtwoturn{\vsat} = 1$, then $\sat$ has a winning pure strategy, thus wins with probability $1$ if she applies it in the concurrent game. Hence $\valconcurrent{\vsat}=1$.
	
	Now suppose $\valtwoturn{\vsat}=0$, then for every pure strategy $\sigma$ of $\sat$, $\unsat$ has a winning strategy against $\sigma$. As a result, by taking a uniform distribution over its strategies, $\unsat$ can achieve a positive probability to win. As a result, $\valconcurrent{\vsat} < 1$.	
\end{proof}

\begin{remark}
	We can make a similar statement about the dual of the two-turn CTL game, described in Remark \ref{dualCTLgame}. For all sets of states and specifications the value given by the dual game is $0$ if and only if the concurrent value is.
\end{remark}

\begin{proposition}[2-turn versus concurrent importance values]\label{prop:comparison}
	Let $S_1, \ldots, S_n$ be a partition of the states of an MTS. If a set of states $S_i$ is useful with respect to the 2-turn \Cref{2turnAsymetricCTL}, then it is useful with respect to the concurrent \Cref{NashImportance}.
\end{proposition}

\begin{proof}
	Suppose $\importanceconcurrent{S_i}=0$. Then for all $J \subseteq \set{1,\ldots,n}$ we have $\valconcurrent{\bigcup_{j \in J} S_j} = \valconcurrent{\bigcup_{j \in J \cup \set{i}} S_j}$. Then by \Cref{Concurrentdeterminestwoturn}, for all $J$ we have \[\valtwoturn{\bigcup_{j \in J} S_j} = \valtwoturn{\bigcup_{j \in J \cup \set{i}} S_j}\] and thus $\importancetwoturn{S_i}=0$.
\end{proof}
The converse of \Cref{prop:comparison} does not hold as shown by the following example.
\begin{example}\label{ex:CTL comparison}
We consider the MTS displayed in \Cref{example11} and the formula $\phi_1 \lor \phi_2 \lor \phi_3$, with:
\begin{align*}
	&\phi_1 = EX(b \land EX c) \land AX (\neg c \land \neg (a \land EX c))\\
	&\phi_2 = AX a \land EXEX c\\
	&\phi_3 = EX c \land EX(b \land EX c) \land EX(a \land EX c)
\end{align*}
In this system, $\phi_1$ expresses that the only path from $0$ to $3$ is through $2$, $\phi_2$ that the only path is through $1$ and $\phi_3$ that all three paths exist. 
	
\begin{figure}[H]
		\centering
		\begin{tikzpicture}[xscale=1.7,yscale=1.3,AUT style]
\node[state,initial, scale=\myscale] (0) at (-1.2,0) {$0$};
\node[state, scale=\myscale] (1) at (-0.6,1) {$1$};
\node[state, scale=\myscale] (2) at (-0.6,-1) {$2$};
\node[state, scale=\myscale] (3) at (0,0) {$3$};

\node (0') at (-0.9,0.2) {\color{blue}$\emptyset$};
\node (1') at (-0.9,1.2) {\color{blue}$\set{a}$};
\node (2') at (-0.9,-1.3) {\color{blue}$\set{b}$};
\node (3') at (0.3,-0.4) {\color{blue}$\set{c}$};

\node (A) at (1.5,1.2) {$\importanceconcurrent{0} = 7/12$};
\node (B) at (1.5,0.9) {$\importanceconcurrent{1} = 1/3~~$};
\node (C) at (1.5,0.6) {$\importanceconcurrent{2} = 1/12$};
\node (D) at (1.5,0.3) {$\importanceconcurrent{3} = 0~~~~$};

\node (A) at (1.5,-0.3) {$\importancetwoturn{0} = 1/2$};
\node (B) at (1.5,-0.6) {$\importancetwoturn{1} = 1/2$};
\node (C) at (1.5,-0.9) {$\importancetwoturn{2} = 0~~$};
\node (D) at (1.5,-1.2) {$\importancetwoturn{3} = 0~~$};

\path[->] (0) edge node[right] {} (1);
\path[->, dashed] (0) edge node[right] {} (2);
\path[->, dashed] (1) edge node[right] {} (3);
\path[->, dashed] (2) edge node[right] {} (3);
\path[->,dashed, bend right] (0) edge node[above] {} (3);
\path[->, loop right] (3) edge node[right] {} (3);

\end{tikzpicture}
		\caption{MTS of Ex. \ref{ex:CTL comparison}}
		\label{example11}
\end{figure}
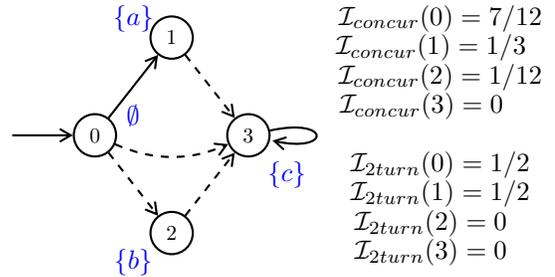

The computation of the concurrent game importance values is lengthy, but straightforward. We observe that $\sat$ has a pure winning strategy whenever she has states $0$ and $1$, and $\unsat$ has a pure winning strategy whenever he has either $0$ or $1$ and $2$. The remaining case is when $\sat$ has $0$ and $2$ and $\unsat$ has $1$, so $\sat$ can choose to allow or not the paths $0,3$ and $0,2,3$, and $\unsat$ can choose to allow or not path $0,1,3$. 

Then one can observe that the case where $\unsat$ allows path $0,1,3$ with probability $1/2$ and $\sat$ never allows $0,3$ and allows $0,2,3$ with probability $1/2$ is a Nash equilibrium, thus the set $\set{1,3}$ has value $1/2$. 
This example shows that in some cases some sets of states may be useless from the 2-turn CTL point of view but not from the Concurrent CTL one.
\end{example}

\section{Conclusion}

We have introduced a new measure of the influence that a part of a system has on whether a given specification is satisfied.
We studied it in the context of two model-checking frameworks, \ltl formulas against Kripke structures and \ctl formulas against modal transition systems.
In most of the cases we provided tight complexity bounds for the corresponding computational problems.
A general conclusion is that the notion of importance value is natural, but still costly in terms of complexity, especially in the case of \ltl. This problem can be mitigated by considering sets of states rather than single states, and formulas from weaker logics. 

We expect that the principle of designing a game and computing the importance of a part of the system by shifting its control from one player to the other can be easily adapted to many model-checking problems. We have studied here classical and basic logics, but one could try to find or design logics more well-suited to the computation of the importance, yielding lower complexities. 

Another continuation of this work would be a fairer definition of the importance in the case of CTL model-checking. Some subsets of CTL formulas may allow us to design a game in which the players can simultaneously choose transitions on the structure and prove the formula without disadvantaging one of the two. This could be related to the notion of good-for-games automata. 

Finally we can extend the definition of value to probabilistic games, by defining the value as the maximal probability of success that $\sat$ can achieve. This gives us a natural notion of importance in probabilistic games that calls for study.

\bibliographystyle{myIEEEtran}
\bibliography{bibliography}

\clearpage
\appendix

\subsection{General properties}

\textbf{Proof of \Cref{ComplementGame}. }
Let $\overline{\game{}}$ be the game with the same arena and initial state as $\game{}$, but the complement objective $\overline{\Omega} = S^{\omega} \setminus \Omega$. 
Then for all $1 \leq i \leq n$, the importance of $S_i$ is the same for games $\game{}$ and $\overline{\game{}}$.

\begin{proof}
	For all $S' \subseteq S$ let $\valbar{S'}$ be the value of $S'$ in $\overline{\game{}}$ and let $\importancebar{S_i}$ be the importance of $S_i$ in $\overline{\game{}}$ for all $i$.
	For all permutations $\pi \in \Pi_n$ let $\tilde{\pi}$ be the mirror permutation, such that for all $1 \leq i \leq n$, $\tilde{\pi}(i) = \pi(n+1-i)$.
	As the function associating its mirror to each permutation is a bijection from $\Pi_n$ to itself, we can rewrite $\importance{S_i}$ as
	\[ \importance{S_i} = \frac{1}{n!} \sum_{\pi \in \Pi_n} \val{S^{\tilde{\pi}}_{\geq i}} - \val{S^{\tilde{\pi}}_{\geq i} \setminus S_i} \]
	As $S^{\tilde{\pi}}_{\geq i} = S \setminus (S^{\pi}_{\geq i} \setminus S_i)$, we have
	\[\val{S^{\tilde{\pi}}_{\geq i}} = 1 - \valbar{S^{\pi}_{\geq i} \setminus S_i}\] ($\sat$ wins with states $S^{\tilde{\pi}}_{\geq i}$ for objective $\Omega$ if and only if she loses with states $S \setminus S^{\tilde{\pi}}_{\geq i}$ for objective $\overline{\Omega}$). 
	Similarly, we have \[\val{S^{\tilde{\pi}}_{\geq i} \setminus S_i} = 1 - \valbar{S^{\pi}_{\geq i}}\].	
	
	As a result, for all $\pi \in \Pi_n$, \[\valbar{S^{\pi}_{\geq i}} - \valbar{S^{\pi}_{\geq i} \setminus S_i} = \val{S^{\tilde{\pi}}_{\geq i}} - \val{S^{\tilde{\pi}}_{\geq i} \setminus S_i}\]
	Finally, we obtain 
	\[ \importance{S_i} = \frac{1}{n!} \sum_{\pi \in \Pi_n} \valbar{S^{\pi}_{\geq i}} - \valbar{S^{\pi}_{\geq i} \setminus S_i} = \importancebar{S_i} \qedhere\]
\end{proof}

\subsection{LTL Proofs}

\textbf{Proof of Theorem~\ref{2exptimeltlimportance}. }
The usefulness and importance threshold problems for LTL with respect to Kripke structures are \twoexpt-complete.

\begin{proof}
	First, as one can solve LTL games in doubly exponential time, one can compute the value of any subset of the states of $\kripke$ in doubly exponential time as well. There are exponentially many such subsets, thus the computation of all those values takes again doubly exponential time. The computation of the importance then comes down to enumerating orderings of the states and computing the sum along the way. As a result, one can compute the importance and compare it with $\tau$ in doubly exponential time, thus the importance threshold problem (and thus also the usefulness one) is in \twoexpt.
	
	For the hardness, we prove that the usefulness problem is \twoexpt-hard in the case where the states are partitioned in singletons. The hardness of the usefulness and importance threshold problems follow directly. We reduce the problem of solving LTL games. Let $\kripke= (S, \ap, \Delta, init, \lambda)$ be a Kripke structure, let $\phi$ be an LTL formula, and let $\vsat \sqcup \vunsat = S$ be a partition of $S$ between states of $\sat$ and $\unsat$. We consider the LTL game $\game{}$ induced by those parameters.
	
	Consider the Kripke structure $\kripke' = (S', \ap', \Delta', c_s, \lambda')$ with $S' = S \cup \set{c_s, c_u, sink, t}$, $\ap' = \ap \cup S'$, and

	\begin{align*}
		\Delta' = \Delta \cup &\set{(s, sink) \mid s \in S'}\\
		\cup &\set{(c_s, s) \mid s \in \vunsat} \cup \set{(c_u, s) \mid s \in \vsat}\\
		\cup &\set{(c_s,c_u), (c_u, t), (t,init)}
	\end{align*}
	and for all $s \in S'$, $\lambda'(s) = \lambda(s) \cup \set{s}$ if $s \in S$ and $\lambda'(s) = \set{s}$ otherwise. In other words, every state is labeled with its own name. See Figure~\ref{figureltl} for an illustration of the construction.
	
	Let $\phi' = \neg \phi_{checkUnsat} \lor ( \phi_{checkSat} \land X^3 \phi)$ with 
	
	\begin{align*}
		\phi_{checkSat} = &\neg X sink \\
		\land & X \neg c_u \Rightarrow X^2 sink\\
		\land & X^2 t \Rightarrow X^3init\\
		\land &X^3 G (\bigvee_{s \in \vsat} s \Rightarrow \neg X sink)
	\end{align*}
	
	\begin{align*}
		\phi_{checkUnsat} = &[X c_u \Rightarrow ( \neg X^2 sink \land (\neg X^2 t \Rightarrow X^3 sink))]  \\
		\land & X^3 G (\bigvee_{s \in \vunsat} s \Rightarrow \neg X sink)
	\end{align*}

	This construction can be done in logarithmic space. The intuition is that if some state in $\vunsat$ belongs to $\sat$, then she can win by going from $c_s$ to that state and then to $sink$. Similarly if some state of $\vsat$ belongs to $\unsat$, then he can win by going to that state from $c_u$ and then to $sink$. In both cases players win without using $t$. The remaining case is when $\sat$ owns states of $\vsat$ and $\unsat$ of $\vunsat$. Then if $\unsat$ owns $t$, he can win by going from there to $sink$, otherwise the players have to play the original game $\game{}$ from $init$. As a result, $t$ is useful if and only if $\sat$ wins $\game{}$.
	We will now prove that the state $t$ is useful with respect to $\phi'$ if and only if $\sat$ wins the original LTL game. 	
	
	First suppose that $\sat$ wins $\game{}$, then we consider \mbox{$T=\set{c_s} \cup \vsat$}. Player $\sat$ loses with $T$:
	
	\begin{itemize}
		\item If she goes from $c_s$ to $sink$ she loses.
		
		\item If she goes from $c_s$ to a state of $\vunsat$, $\unsat$ can then go to some state different from $sink$ and not satisfy $\phi_{checkSat}$ while satisfying $\phi_{checkUnsat}$ (recall that in our definition of Kripke structure we assume every state to have at least one outgoing transition).  
		
		\item If she goes from $c_s$ to $c_u$, then $\unsat$ can go to $t$ then $sink$ and not satisfy $\phi_{checkSat}$ while satisfying $\phi_{checkUnsat}$.  
	\end{itemize}
	
	Moreover, player $\sat$ wins with $T \cup \set{t}$, as she can start by going from $c_s$ to $c_u$ and:
	
	\begin{itemize}
		\item If $\unsat$ goes to $sink$ from $c_u$, he loses.
		
		\item If $\unsat$ goes from $c_u$ to a state of $\vsat$, $\sat$ can then go to some state different from $sink$ and not satisfy $\phi_{checkUnsat}$.  
		
		\item If $\unsat$ goes from $c_u$ to $s$, then $\sat$ can go to $init$ and then win by playing a winning strategy for $\game{}$, thus satisfying $\phi_{checkSat} \land X^3 \phi$.  
	\end{itemize}
	
	Thus $t$ is useful.
	
	Now suppose that $t$ is useful, let $T \subseteq S'$ be a set of states such that $(t,T)$ is critical. $T$ has to contain $c_s$ as otherwise $\unsat$ can go from $c_s$ to $sink$ directly and make $\sat$ lose with $T\cup \set{t}$. If $\sat$ had a winning strategy with $T \cup \set{t}$ not going from $c_s$ to $c_u$, then she would also win with just $T$ by applying this strategy as $t$ is then never reached.
	
	As a result, $\sat$ with $T \cup \set{t}$ has to go from $c_s$ to $c_u$. As a consequence, $T$ has to be disjoint from $\vunsat$, as otherwise $\sat$ with $T$ could go from $c_s$ to a state in $T \cap \vunsat$ and from there to $sink$, unsatisfying $\phi_{checkUnsat}$. Further, $c_u$ cannot be in $T$ as otherwise $\sat$ could win by going from $c_u$ to $sink$.
	
	Finally, $\unsat$ cannot win when $\sat$ has $T$ by going from $c_u$ to a state different from $t$ as otherwise he could win when $\sat$ has $T \cup \set{t}$ with the same strategy. As a consequence, $T$ has to contain $\vsat$, as if not $\unsat$ could go from $c_u$ to a state in $\vsat \setminus T$ and then $sink$, winning the game. 
	
	Whether $sink$ is in $T$ is irrelevant to the game as there is only one outgoing transition from $sink$. Thus we can assume that $T = \set{c_s} \cup \vsat$. Suppose $\unsat$ wins $\game{}$, and consider the game where $\sat$ has $T \cup \set{t}$. As $\sat$ has to go from $c_s$ to $c_u$ to win, $\unsat$ can then go from $c_u$ to $s$, and $\sat$ has to go to $init$. Then $\unsat$ can apply his winning strategy for $\game{}$, as $\sat$ loses if she goes to $sink$ and thus cannot go out of $S$.
	This makes $\unsat$ win, contradicting the hypothesis that $\sat$ wins with $T \cup \set{t}$. In conclusion, $\sat$ wins $\game{}$.
	
	As a result, the usefulness and importance threshold problems are \twoexpt-complete for LTL.  
\end{proof}

\textbf{Proof of Theorem~\ref{sharpPcompleteReachability}. }
The importance computation problem for reachability, Büchi, parity and explicit Muller conditions with respect to Kripke structures are \sharppoly-complete.

\begin{proof}
	First, the upper bound for reachability, Büchi and explicit Muller conditions is obtained by constructing a Turing machine guessing an ordering of $\set{1, \ldots, n}$, and accepting if the set $J$ of indices coming after $i$ in the ordering is such that $(i,J)$ is critical, which can be checked in polynomial time. The number of accepting runs is the number of permutations satisfying this condition, i.e., $n!\importance{S_i}$. The problem is therefore in \sharppoly.
	
	For parity conditions, we rely on the result by Jurdzi{\' n}ski that solving parity games can be done by a polynomial-time unambiguous Turing machine, i.e., a nondeterministic machine that has at most one accepting run on every input~\cite{Jurdzinski98}.
	
	This allows us to build a machine that takes as input a Kripke structure $\kripke$, a partition $S_1, \ldots, S_n$ of the states, an index $i$ and a coloring $c$ and guesses an ordering of $\set{1,\ldots,n}$. Let $J$ be the set of indices coming after $i$ in the permutation, our machine can simulate the unambiguous Turing machine in order to check that $\sat$ wins with $\bigcup_{j \in J \cup \set{i}} S_j$ and $\unsat$ wins with $\bigcup_{j \in J} S_j$.
	The number of accepting runs of this machine is precisely $n!\importance{S_i}$.
	
	Our reduction to show \sharppoly-hardness is from the problem of counting solutions to a 1-in-3SAT instance, i.e., given a 3SAT formula, counting the number of valuations such that every clause has exactly one satisfied literal. This problem was shown to be \sharppoly-complete by Creignou and Hermann~\cite{CreignouHermann1996}.
	
	Let $\phi = C_1 \land C_2 \land \cdots \land C_k$ be a 3SAT formula, with $C_j = (\ell_j^1 \lor \ell_j^2 \lor \ell_j^3)$ for all $j$, and let $\set{x_1, \ldots, x_n}$ be the set of variables appearing in $\phi$. We first construct the formula $\psi = \bigwedge_{j=1}^k C_j \land \bigwedge_{j=1}^n (x_j \lor \neg x_j) \land \bigwedge_{j=1}^k C_j^{1,2} \land C_j^{2,3} \land C_j^{3,1}$ with $C_j^{i_1, i_2} = (\neg \ell_j^{i_1} \lor \neg \ell_j^{i_2})$.
	
	One can check that a valuation $\nu : \set{x_1,\ldots,x_n} \to \set{\bot,\top}$ satisfies $\psi$ if and only if it satisfies exactly one literal per clause in $\phi$.
	
	Furthermore if a valuation satisfies $\psi$, then it satisfies exactly one literal in every clause except for exactly one of $C_j^{1,2},C_j^{2,3},C_j^{3,1}$ for each $j$, in which it satisfies both literals.  
	
	We reuse the construction from the proof of Proposition~\ref{usefulnessReach}, with $\psi$ as our 3SAT instance. Recall that this construction used a reachability condition, easily expressible as a Büchi, parity or Muller condition, making the reduction work for all those winning conditions. As $sink$ and $f$ only have one outgoing transition, they have no influence on the satisfaction of a specification by a set of states, thus their importance is $0$. As a consequence, by Corollary~\ref{UselessStatesCorollary} they can be ignored in the computation of the importance, thus we will only consider set of states containing neither.
	 Then a team of states $T$ makes $(s,T)$ critical if and only if it contains all the $c_i$ but not $s$ and there exists a valuation $\nu$ satisfying $\psi$ such that $T$ contains exactly the states associated literals satisfied by $\nu$, except in clauses $C_j^{i_1,i_2}$ in which $\nu$ satisfies both literals, in which $T$ contains either one of the two states or both.
	
	As a result, for every valuation $\nu$ satisfying $\psi$, we have exactly $3^k$ sets of states $T$ making $s$ critical and matching that valuation. Indeed, $T$ is completely determined by $\nu$ except for one $C_j^{i_1,i_2}$ for each $1\leq j \leq k$, in which it has three possibilities: contain the first literal, the second, or both. 
	
	A brief analysis shows that for each such valuation $\nu$, there are, for each $0 \leq i \leq k$, $\binom{i}{k}2^{k-i}$ corresponding teams of size $2i+k-i+2k+2n$ (those teams being the ones containing both literals in $i$ out of the $k$ clauses $C_j^{i_1,i_2}$ in which $\nu$ satisfies both literals), adding up to $3^k$ teams.
	
	Let $N$ be the total number of states in the Kripke structure. By Corollary~\ref{UselessStatesCorollary}, the number of valuations satisfying $\phi$ with exactly one satisfied literal per clause is therefore $\frac{(N-2)!}{N! M} N!I(s)$, with 
	\[M = \sum_{i=1}^k \binom{i}{k}2^{k-i} (i+3k+2n)! ([N-2]-i-3k-2n)!\]  
	
	As $M$ can be computed in polynomial time, the problem is therefore \sharppoly-complete.
\end{proof}

\textbf{Proof of Theorem~\ref{EmersonLei}. }
The value, usefulness and importance threshold problems for Emerson-Lei conditions are \pspace-complete.
	
\begin{proof}
	
	The upper bounds arise from the complexity of solving Emerson-Lei games, which are \pspace-complete~\cite{HunterDawar2005}. As enumerating permutations of the states can be done in linear space, one can compute the importance of a set of states in \pspace.
	
	For the lower bounds, we adapt a classic proof that Emerson-Lei games are \pspace-hard to our framework. We only need to prove that the usefulness problem is \pspace-hard as the importance threshold problem reduces to it. Further, we only use the particular case when the set of states is partitioned in singletons. 
	
	We reduce the QSAT problem. Let $Q_1 x_1 \cdots Q_k x_k \psi$ be a QSAT instance, we consider the following Kripke structure:
	\begin{itemize}
		\item $\set{c_i,x_i, \neg x_i \mid 1 \leq i \leq k} \cup \set{s, win_S, win_U}$ is the set of states, $c_1$ is the only initial state.
		
		\item For all $1 \leq i \leq k$ there are transitions $(c_i,x_i), (c_i,\neg x_i), (x_i, c_{i+1}), (\neg x_i, c_{i+1})$, with $c_{k+1} = s$. There is also a transition $(c_i, win_U)$ if $Q_i = \exists$ and $(c_i, win_S)$ if $Q_i = \forall$. The remaining transitions are $(s,c_1), (s,win_U), (win_S,win_S), (win_U,win_U)$.
	\end{itemize}
	The labeling is irrelevant here. Figure \ref{figureEmersonLei} illustrates the construction.
	
	\begin{figure}
		\begin{tikzpicture}[xscale=1.7,yscale=1.5,AUT style]
\node[state,initial, scale=\myscale] (c1) at (0,0) {$c_1$};
\node[state, scale=\myscale] (x1) at (1,0.5) {$x_1$};
\node[state, scale=\myscale] (nx1) at (1,-0.5) {$\neg x_1$};
\node[state, scale=\myscale] (x2) at (3,0.5) {$x_2$};
\node[state, scale=\myscale] (nx2) at (3,-0.5) {$\neg x_2$};
\node[state, scale=\myscale] (c2) at (2,0) {$c_2$};
\node[state, scale=\myscale] (s) at (4,0) {$s$};
\node[state, scale=\myscale] (ws) at (0,-1) {$win_S$};
\node[state, scale=\myscale] (wu) at (2,-1) {$win_U$};

\path[->] (c1) edge node[right] {} (x1);
\path[->] (c1) edge node[right] {} (nx1);
\path[->] (c2) edge node[right] {} (x2);
\path[->] (c2) edge node[right] {} (nx2);
\path[->] (x1) edge node[above] {} (c2);
\path[->] (nx1) edge node[above] {} (c2);
\path[->] (x2) edge node[above] {} (s);
\path[->] (nx2) edge node[above] {} (s);
\path[->, bend right=60] (s) edge node[left] {} (c1);
\path[->] (c1) edge node[right] {} (ws);
\path[->] (c2) edge node[above] {} (wu);
\path[->, bend left = 50] (s) edge node[above] {} (wu);
\path[->, loop left] (ws) edge node[right] {} (ws);
\path[->, loop left] (wu) edge node[above] {} (wu);
\end{tikzpicture}
		\caption{Kripke structure corresponding to formula $\forall x_1, \exists x_2, x_1 \land \neg x_2$}
		\label{figureEmersonLei}
	\end{figure}
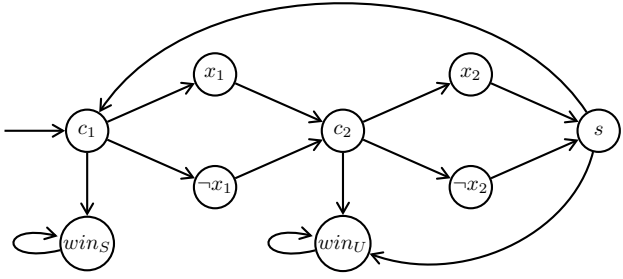
	
	For all $1 \leq i \leq k$ let \[ \phi_i = Inf(x_i) \land Inf(\neg x_i) \land \bigwedge_{j=1}^i \neg(Inf(x_j) \land Inf(\neg x_j))\] expressing that $i$ is the minimal $i$ such that both $x_i$ and $\neg x_i$ are visited infinitely many times.
	
	We take as winning condition for $\sat$ the formula \[(\psi' \lor Inf(win_S) \lor \bigvee_{Q_i = \forall} \phi_i) \land \neg Inf(win_U) \land \bigwedge_{Q_i=\exists} \neg \phi_i\] where $\psi'$ is $\psi$ in which every $x_i$ has been replaced with $Inf(x_i)$.
	
	This construction can be done in logarithmic space. We will now prove that the QSAT formula is valid if and only if state $s$ is useful.
	
	Suppose the QSAT formula is valid, let $T = \set{c_i \mid Q_i= \exists}$. Clearly $\sat$ loses with $T$ as by taking the transition to $win_U$ from $s$, $\unsat$ can guarantee that every play reaches $win_U$ and thus wins.
	
	As the QSAT formula is valid, there exist functions $(f_i)_{Q_i = \exists}$ such that for all $i$ $f_i : \set{\top, \bot}^{i-1} \to \set{\top, \bot}$ and for all $\nu : \set{x_1, \ldots, x_k} \to \set{\top, \bot}$ such that for all $f_i$ we have $\nu(x_i) = f_i(\nu(x_1),\ldots, \nu(x_{i-1}))$, $\nu$ satisfies $\psi$.
	
	Further, as $\sat$ makes all the existential choices, if $\sat$ chooses according to $f_i$ from every $c_i$ she owns, and takes the transition to $c_1$ from $s$. Suppose $\sat$ takes the transitions to $x_i$ and $\neg x_i$ infinitely many times, then as $\sat$ plays according to functions $f_i$, it means there exists a $j<i$ such that $x_j$ and $\neg x_j$ were visited infinitely many times.
	
	As a consequence, the minimal $i$, if it exists, such that $x_i$ and $\neg x_i$ are visited infinitely many times is such that $Q_i = \forall$. If it exists, then $\phi_i$ is satisfied, while $\phi_j$ is not satisfied for any other $j$, and as $win_U$ is never visited, $\sat$ wins. 
	
	If it does not exist, then for all $j$ exactly one of $x_j, \neg x_j$ is visited infinitely many times, and as $\sat$ plays according to the $f_i$, we have that $\psi'$ is satisfied. As no $\phi_j$ is satisfied and $win_U$ is never visited, $\sat$ wins.
	
	Now suppose the QSAT formula is not satisfiable, and suppose there exists $T$ such that $\sat$ wins with $T \cup \set{s}$ but not with $T$. If there exists $c_i \in T$ such that $Q_i = \forall$ or $c_i \notin T$ such that $Q_i = \exists$, then either $\sat$ can reach $win_S$ and win with $T$, or $\unsat$ can reach $win_U$ and win while $\sat$ has $T\cup \set{s}$. Whether $win_S, win_U$ or the $x_i, \neg x_i$ belong to $T$ is irrelevant as they have only one outgoing transition. 
	
	Thus we can assume that $T = \set{c_i \mid Q_i= \exists}$. By similar arguments as above, there exist functions $(f_i)_{Q_i = \forall}$ such that $f_i$ associates to the $i-1$ values of the previous literals a valuation of $x_i$, and any valuation respecting those functions does not satisfy $\psi$. And again by similar arguments as above, playing according to those functions allows $\unsat$ to win the game while $\sat$ has $T\cup \set{s}$, contradicting the hypothesis that $\sat$ wins with $T\cup \set{s}$. As a result, $s$ is not useful.
\end{proof}

\textbf{Proof of Proposition~\ref{usefulnessRabin}. }
The usefulness problem for Rabin conditions is \mbox{\sigmatwo-complete}.

\begin{proof}
	For the upper bound we simply consider a nondeterministic Turing machine guessing a set of indices $J$ and calling an \np oracle twice to check that player $\sat$ wins with $\bigcup_{j \in J \cup \set{i}} S_j$ as set of states and loses with just $\bigcup_{j \in J} S_j$.
	
	For the lower bound, we reduce the dual of the $\forall \exists$3SAT problem, known to be \pitwo-complete~\cite{Stockmeyer1976}\cite{Wrathall1976}.
	Given a formula $\phi = \forall (x_i)_{1\leq i \leq n}, \exists (y_i)_{1 \leq i \leq p} \psi$ with $\psi = \bigvee_{i=1}^k Cl_i$ a $\forall \exists$3SAT instance, we are going to construct a Kripke structure $\kripke$ (with states partitioned in singletons), a state $s$ and a Rabin condition $R$ such that $s$ is useful to $\kripke$ with respect to $R$ if and only if this formula is \emph{not} valid.
	
	First of all note that we can assume that every clause contains an existential variable $y_i$. Indeed, any clause $(\ell_1 \lor \ell_2 \lor \ell_3)$ can be replaced by $(\ell_1 \lor \ell_2 \lor y) \land (\neg y \lor \ell_3) \land (\neg \ell_3 \lor y)$, with $y$ a fresh variable which we add to the set of existential ones. One can check that we obtain a formula equisatisfiable to the previous one.\\~\\
	
	Consider the structure $\kripke$ whose states are elements of 
	\begin{align*}
		&\set{c_i, x_i, \neg x_i, c'_i, x'_i, \neg x'_i \mid 1 \leq i \leq n}\\
		\cup &\set{sk_{x_i}, sk_{\neg x_{i}}, ret_{x_{i}}, ret_{\neg x_{i}} \mid 1 \leq i \leq n}\\
		\cup &\set{y_j, \neg y_j \mid 1 \leq j \leq p} \cup \set{Cl_i \mid 1 \leq i \leq k} \cup \set{s,sink}
	\end{align*}
	
	whose initial state is $c_1$ and whose transitions are as follows:
	
	\begin{itemize}
		
		\item There are transitions from $init$ to itself, to $c_1$ and to every $Cl_j$. 
		
		\item For all $i$ there are transitions from $c_i$ to $x_i$ and $\neg x_i$ and from $x_i$ and $\neg x_i$ to $c_{i+1}$, with $c_{n+1}=c'_1$.
		
		\item For all $i$, for all $\ell \in \set{x_i, \neg x_i}$, there are transitions $(\ell, sk_\ell), (sk_\ell, \neg \ell')$ 
		
		\item We have transitions $(c'_i, x'_{i}), (c'_i, \neg x'_{i}), (x'_i, c'_{i+1}), (\neg x'_i, c'_{i+1})$ for all $1 \leq i \leq n$, with the convention $c'_{n+1} = s$.
		
		\item  For all $\ell$ of the form $x_i$ or $\neg x_i$, for all clause $Cl_j$ containing $\ell$ there are transitions $(Cl_j, ret_\ell)$ and $(ret_\ell, \ell)$.
		
		\item For all $\ell$ of the form $y_i$ or $\neg y_i$, for all $Cl_j$ containing $\ell$, there is a transition $(Cl_j, \ell)$ and a transition $(\ell,c_1)$. 
		
		\item For all clause $Cl_j$ there is a transition $(s,Cl_j)$.
		
		\item There are transitions from all $c_i, c'_i, x_i, \neg x_i, x'_i, \neg x'_i, Cl_j$ to $sink$. 
	\end{itemize}

	Figure~\ref{figureRabin} illustrates the construction. There are transitions from the blue and white states to $skip$, which is omitted on the picture. The blue states are the ones hardcoded to belong to $\sat$, the grey ones are the ones that have only one outgoing transition, and the white ones are the ones encoding the valuation of the $x_i$.
	
	As states $sink, ret_\ell, sk_\ell, y_i, \neg y_i$ have only one outgoing transition, whether they belong to $\sat$ or $\unsat$ has no consequence on the game. In the proof that follows we will ignore which player they belong to. 
	
	We take as Rabin condition 
	\begin{align*}
		R = &\set{(\set{y_i}, \set{\neg y_i}), (\set{\neg y_i}, \set{y_i}) \mid 1 \leq i \leq p}\\
		\cup &\set{(\set{sk_\ell}, \emptyset), (\set{ret_\ell}, \emptyset) \mid 1\leq i \leq n, \ell \in \set{x_i, \neg x_i}}\\
		\cup &\set{(\emptyset, \set{c'_1, sink}), (\set{init},\emptyset)}.
	\end{align*}
	
	The construction can be done in logarithmic space. We will now show that the formula $\phi$ is not valid if and only if $s$ is useful in the Kripke structure with respect to this Rabin condition.
	
	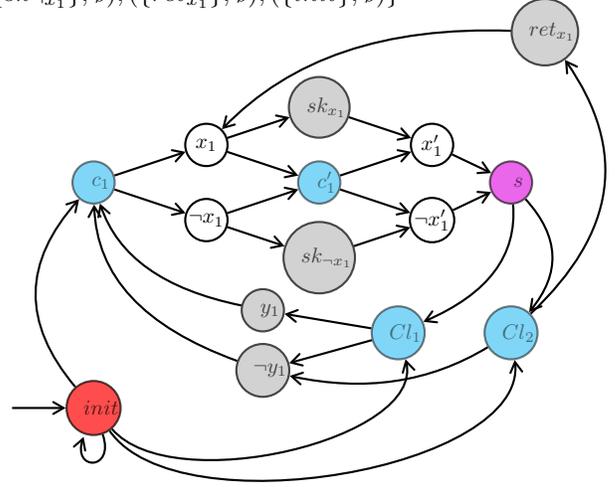
\begin{figure}[H]
		\begin{center}
			\begin{tikzpicture}[xscale=1.5,yscale=1,AUT style]
\node (text1) at (1,3) {$\{(\set{x_2},\set{\neg x_2}),(\set{\neg x_2}, \set{x_2}), (\set{sk_{x_1}}, \emptyset),$};
\node (text2) at (0.85,2.5) {$(\set{sk_{\neg x_1}}, \emptyset), (\set{ret_{x_1}}, \emptyset), (\set{init}, \emptyset)\}$};
\node[state, initial, fill = red,opacity=.7,text opacity=1, scale=\myscale] (init) at (0,-3) {$init$};
\node[state, fill=cyan,opacity=.5,text opacity=1, scale=\myscale] (c1) at (0,0) {$c_1$};
\node[state, scale=\myscale] (x1) at (1,0.5) {$x_1$};
\node[state, scale=\myscale] (nx1) at (1,-0.5) {$\neg x_1$};
\node[state, fill=cyan,opacity=.5,text opacity=1, scale=\myscale] (c'1) at (2,0) {$c'_1$};
\node[state, scale=\myscale] (x'1) at (3,0.5) {$x'_1$};
\node[state, scale=\myscale] (nx'1) at (3,-0.5) {$\neg x'_1$};
\node[state, fill = lightpurple, opacity=.6,text opacity=1, scale=\myscale] (s) at (3.7,0) {$s$};
\node[state, fill=cyan,opacity=.5,text opacity=1, scale=\myscale] (C1) at (2.7,-2) {$Cl_1$};
\node[state, fill=cyan,opacity=.5,text opacity=1, scale=\myscale] (C2) at (3.7,-2) {$Cl_2$};
\node[state, fill=light-gray,opacity=.7,text opacity=1, scale=\myscale] (r12) at (4,2) {$ret_{x_1}$};
\node[state, fill=light-gray,opacity=.7,text opacity=1, scale=\myscale] (skip1) at (2,1) {$sk_{x_1}$};
\node[state, fill=light-gray,opacity=.7,text opacity=1, scale=\myscale] (skipn1) at (2,-1) {$sk_{\neg x_1}$};
\node[state, fill=light-gray,opacity=.7,text opacity=1, scale=\myscale] (x2) at (1.5,-1.7) {$y_1$};
\node[state, fill=light-gray,opacity=.7,text opacity=1, scale=\myscale] (nx2) at (1.5,-2.5) {$\neg y_1$};

\path[->,bend left = 30] (init) edge node {} (c1);
\path[->,bend right = 80] (init) edge node {} (C1);
\path[->,bend right = 80] (init) edge node {} (C2);
\path[->, loop below] (init) edge node {} (init);
\path[->] (c1) edge node[right] {} (x1);
\path[->] (c1) edge node[right] {} (nx1);
\path[->] (c'1) edge node[right] {} (x'1);
\path[->] (c'1) edge node[right] {} (nx'1);
\path[->,bend left] (x2) edge node[right] {} (c1);
\path[->, bend left] (nx2) edge node[right] {} (c1);
\path[->] (x1) edge node[right] {} (c'1);
\path[->] (nx1) edge node[above] {} (c'1);
\path[->] (x'1) edge node[right] {} (s);
\path[->] (nx'1) edge node[above] {} (s);
\path[->, bend left] (s) edge node[left] {} (C1);
\path[->,bend left] (s) edge node[left] {} (C2);
\path[->, bend right] (C2) edge node[left] {} (r12);
\path[->, bend right] (r12) edge node[left] {} (x1);
\path[->] (C1) edge node[left] {} (x2);
\path[->] (C1) edge node[left] {} (nx2);
\path[->,bend left] (C2) edge node[left] {} (nx2);
\path[->] (x1) edge node[right] {} (skip1);
\path[->] (nx1) edge node[right] {} (skipn1);
\path[->] (skip1) edge node[right] {} (x'1);
\path[->] (skipn1) edge node[right] {} (nx'1);
\end{tikzpicture}
			\caption{Construction for the formula $\forall x_1 \exists x_2 (x_2 \lor \neg x_2) \land (x_1 \lor \neg x_2)$. The sink state is omitted. The Rabin condition is displayed at the top of the figure.}
			\label{figureRabin}
		\end{center}
	\end{figure}
	
	Suppose that $\phi$ is not valid, let $\nu$ be a valuation of the $x_i$ not satisfying $\exists (y_i) \psi$. We take as set of states $T$ all the $c_i, c'_i$, all the $Cl_j$, and the $\ell$ and $\ell'$ such that $\nu(\ell) = \top$. 
	
	If $\sat$ has states $T \cup \set{s}$, then she can pick any clause $Cl_j$, any literal $\ell$ in $Cl_j$ of the form $y_i$ or $\neg y_i$. 
	
	If $\unsat$ loops on $init$ forever, then $\sat$ wins. If he chooses to go to $c_1$, then as there is a path $P$ in $T$ from $c_1$ to $s$, $\sat$ can repeat indefinitely the cycle taking $P$ from $c_1$ to $s$, then going through $Cl_j$, then $\ell$, then back to $c_1$. This allows $\sat$ to win as she goes through $\ell$ infinitely many times without going through $\neg \ell$.
	
	If $\unsat$ goes to some $Cl_j$ from $init$ then $\sat$ can simply go to some $y_i$ or $\neg y_i$ (recall that we assumed every clause to contain an $y_i$ or $\neg y_i$), then to $c_1$ and from there play as in the previous case.
	
	If $\sat$ has states $T$, then we proceed by contradiction. Suppose $\sat$ has a winning strategy, then as she is the player with a Rabin winning condition, she has a positional one\cite{Klarlund1994}. In particular from every $Cl_j$ $\sat$ picks either a successor $ret_\ell$ with $\ell$ in $Cl_j$ or a successor $y_i$ or $\neg y_i$ in $Cl_j$. In the first case, $\ell$ has to be satisfied by $\nu$, otherwise after $ret_\ell$ the game reaches $\ell$, from which $\unsat$ goes to $sink$ and wins.
	
	Further, from every $c_i$ $\sat$ has to pick the successor $x_i$ or $\neg x_i$ satisfied by $\nu$, otherwise $\unsat$ can then reach $sink$ and win.
	
	From an $x_i$ belonging to $\sat$, she cannot go to $sk_{x_i}$ as then she ends up in $\neg x'_i$, from where $\unsat$ can reach $sink$. The same argument stops her from going to $sk_{\neg x_i}$ from $\neg x_i$. 
	As a result, from $c_1$ $\sat$ has to follow a path to $s$.
	
	As every $ret_\ell$ to which $\sat$ goes from a $Cl_j$ is such that $\nu$ satisfies $\ell$, if there were a valuation of the $y_i$ satisfying every literal $y_i$ or $\neg y_i$ to which $\sat$ goes to from a clause, then we could infer from the strategy of $\sat$ a valuation $\mu$ such that $\nu$ and $\mu$ combined satisfy $\psi$. This would contradict the fact that $\nu$ does not satisfy $\exists (y_i) \psi$, thus there is no such valuation $\mu$.
	
	As a result, there exist $i,j_1, j_2$ such that $\sat$ picks $y_i$ from $Cl_{j_1}$ and $\neg y_i$ from $Cl_{j_2}$. As $\sat$ has to go to $s$ from $c_1$, $\unsat$ can then alternately choose $Cl_{j_1}$ and $Cl_{j_2}$ as successors, thus making $\sat$ go infinitely many times through $y_i$ and $\neg y_i$ (and never through other $y_i$ or $\neg y_i$). Then $\unsat$ wins, contradicting the fact that $\sat$ is playing a winning strategy.
	
	In conclusion, if $\phi$ is not valid, then $s$ is useful for $\kripke$ with respect to $R$.
	
	Now we have to prove that if $s$ is useful for $\kripke$ with respect to $R$, then $\phi$ is not valid. Suppose the former, let $T$ be a set of states such that $\sat$ wins with $T \cup \set{s}$ but not with $T$.
	
	As $\sat$ loses with $T$, $init$ cannot be in $T$, otherwise $\sat$ could win by looping forever on $init$.  
	
	Suppose there exists a $Cl_j \notin T$, then $\unsat$ can win by going to $Cl_j$ from $init$ and then to $sink$, contradicting the fact that $\sat$ wins with $T \cup \set{s}$. Thus $T$ contains all $Cl_j$. If $\unsat$ loops on $init$ forever then $\sat$ wins. If $\unsat$ goes from $init$ to some $Cl_j$ then $\sat$ can go to some $y_i$ or $\neg y_i$ and from there to $c_1$. We can thus assume that the players always end up reaching $c_1$.
	
	As $\sat$ loses with $T$, there cannot be any path in $T$ from $c_1$ to $s$ going through a $sk_{\ell}$, otherwise $\sat$ could go infinitely many times through that $sk_{\ell}$. However as $\sat$ wins with \mbox{$T \cup \set{s}$}, there has to be a path from $c_1$ to $s$ in $T$ (otherwise $\sat$ would have to reach a state of $\unsat$ with a transition to $sink$ and lose). Thus for all $i$, $c_i, c'_i$ belong to $T$, as well as one of $x_i, \neg x_i$ and one of $x'_i, \neg x'_i$. Further, for all $i$, we cannot have both $x_i$ and $\neg x'_i$, or both $\neg x_i$ and $x'_i$ in $T$. As a result for all $i$ either $x_i, x'_i \in T$ and $\neg x_i, \neg x'_i \notin T$ or $\neg x_i, \neg x'_i \in T$ and $x_i, x'_i \notin T$. Let $\nu$ be the valuation of the $x_i$ such that $\nu(x_i)= \top$ if and only if $x_i \in T$.
	
	Let $\mu$ be a valuation of the $y_i$, suppose for the sake of contradiction that the combination of $\nu, \mu$ satisfies $\psi$. Then for all $j$ $Cl_j$ has a transition either to a $ret_\ell$ with $\nu(\ell) = \top$ or to an $y_i$ with $\mu(y_i) = \top$ or to a $\neg y_i$ with $\mu(y_i)=\bot$.
	
	Then by taking from each $Cl_j$ the successor as stated above, $\sat$ wins as she will necessarily either go through a $ret_\ell$ infinitely many times, or through an $y_i$ or $\neg y_i$ infinitely many times while never visiting the opposite literal.
	
	We obtain a contradiction as $\unsat$ is supposed to win the game when $\sat$ only owns $T$.
	
	As a result, $\nu,\mu$ cannot satisfy $\psi$, thus $\phi$ is not valid.
\end{proof}

\textbf{Proof of Theorem~\ref{importanceRabin}. }
The importance computation problem for Rabin conditions is \sharpsigmatwo-complete.

\begin{proof}
	We reduce the problem of counting, given a formula $\psi$ in 3CNF over variables $x_1, \ldots, x_n, y_1, \ldots, y_p$, the number of valuations of the $x_i$ such that for all valuations of the $y_i$, the combination of those does not satisfy $\psi$.
	
	First, let us justify that this problem is \sharpsigmatwo-hard. Let $M$ be a non-deterministic Turing machine with an oracle solving an \np-complete problem (say SAT), let $w$ be an input.
	
	We can assume without loss of generality that $M$ only makes one query to the oracle, and accepts if and only if the answer is negative. 
	
	Indeed, say $M$ has to make queries $\psi_1, \ldots, \psi_k$ to the oracle, all over existential variables $y_1,\ldots, y_p$. It can nondeterministically guess the answers of the oracle and delay the verification to the end of the run. 
	
	Now let us define the order $\leq$ on valuations of the existential variables as the lexicographic order, a valuation $\nu$ being seen as the tuple $(\nu(y_1), \ldots, \nu(y_p))$ and with the convention $\bot \leq \top$. Then for the positive answers $M$ can guess a \emph{minimal} witness valuation for the $y_i$ with respect to $\leq$, negate the formula. The problem, given a SAT instance and a valuation of the existential variables, of checking whether this is the minimal valuation witnessing the satisfiability of the formula, is clearly in \conp. As a result $M$ can guess the minimal valuation of the $y_i$ witnessing the satisfiability of the formula, and then turn it into a SAT formula unsatisfiable if and only if the guess is correct. As there is for every SAT formula a unique minimal valuation satisfying it, $M$ can only make one correct guess, thus its number of runs is unchanged.
	
	Finally, in the end $M$ has to make the oracle check a disjunction of $\exists$ formulas, it can rename variables in order to merge them all into one equivalent SAT instance, and accept if and only if the oracle rejects that formula.
	
	We now use the classical encoding of Turing machines in 3CNF formulas to construct a 3CNF formula $\psi$ over variables $x_1,\ldots,x_n, y_1,\ldots,y_p,r_1, \ldots, r_k$ such that for all valuations of the $x_i, y_i$ (encoding respectively the nondeterministic choices of $M$ and the ones of the oracle), there is a non-accepting run of $M$ on $w$ if and only if there exists a valuation of the $r_i$ (encoding the runs of $M$ and the oracle) satisfying the formula along with these valuations of $x_i, y_i$.
	
	As a result, there is an accepting run of $M$ on $w$ if and only if the formula $\forall (x_i) \exists (y_i), (r_i) \phi_1$ is not valid, and the valuations of the $x_i$ witnessing non-validity are in bijection with the runs of $M$.
	
	Hence the problem is \sharpsigmatwo-hard.
	
	Now in order to prove the hardness for the importance computation problem for Rabin conditions, we use the same construction as in the proof of Proposition~\ref{usefulnessRabin}. Let $\phi$ be a 3CNF formula with $k$ clauses over variables $x_1, \ldots, x_n, y_1, \ldots,y_p$, we consider the Kripke structure from that proof.
	
	Furthermore, a set of states $T$ makes $(s,T)$ critical if and only if it contains the $c_i, c'_i, Cl_i$, and the $x_i, \neg x_i, x'_i \neg x'_i$ encoding a valuation of the $x_i$ such that for all valuation of the $y_i$, the combination of the two valuations does not satisfy $\phi$.
	
	Note that states $sk_\ell, ret_\ell, y_i, \neg y_i$ all have one outgoing transition and thus have importance $0$.
	
	As those sets $T$ all have the same size $k + 4n$, the formula from Corollary~\ref{UselessStatesCorollary} gives us that the number of valuations of the $x_i$ such that for all valuation of the $y_i$, the combination of the two valuations does not satisfy $\phi$ is $\frac{P!}{N!(k+4n)!(P-k+4n)!}N!I(s)$ where $N$ is the number of states in the Kripke structure and $P$ the number of states minus the $sk_\ell, ret_\ell, y_i, \neg y_i$.
	
	As $\frac{P!}{N!(k+4n)!(P-k+4n)!}$ can be computed in polynomial time, the importance computation problem for Rabin conditions is \sharpsigmatwo-complete.
\end{proof}

\subsection{CTL proofs}

\textbf{Proof of Proposition~\ref{2turnCTLvalue}. }
The value problem for two-turn CTL is \mbox{\sigmatwo-complete}.

\begin{proof}
	
	One can reformulate the problem as the existence of a subset of outgoing transitions from $\vsat$ such that for all subsets of outgoing transitions from $\vunsat$, the structure yielded by those subsets of transitions satisfies $\phi$.
	
	As those subsets of transitions are of polynomial size, and as the satisfaction of a Kripke structure by a CTL formula can be checked in polynomial time, the problem is in \sigmatwo.
	
	We now prove the lower bound, by reducing $\exists\forall$SAT. Let $\exists (x_i)_{1 \leq i \leq n}, \forall (y_{i})_{1 \leq i \leq k} \psi$ with $\psi$ quantifier-free be a $\exists\forall$SAT instance. Without loss of generality, we assume that all the negations in $\psi$ have been pushed to the atomic propositions.
	
	We consider the following modal transition system \mbox{$\modal = (S,\ap,\Delta_{must}, \Delta_{may}, init, \lambda)$} with:
	\begin{itemize}
		\item $S = \set{sink} \cup \set{c_i,x_i, \neg{x_i} \mid 1 \leq i \leq n+k}$. The initial state is $c_1$.
		
		\item $\ap = \set{x_i \mid 1 \leq i \leq n+k}$ and $\lambda(x_i) = \set{x_i}$ for all $1\leq i \leq n+k$ and $\lambda(s)= \emptyset$ for all other $s \in S$.
		
		\item $\Delta_{must} = \set{(x_{n+k},sink), (\neg{x_{n+k}}, sink), (sink,sink)} \cup \set{(x_i, c_{i+1}),(\neg x_i, c_{i+1}) \mid 1 \leq i \leq n+k-1}$.
		
		\item $\Delta_{may} = \set{(c_i,x_{i}),(c_i,\neg{x_{i}}) \mid 1 \leq i \leq n+k}$.
	\end{itemize} 
	
	We split $S$ into \[\vsat = \set{sink} \cup \set{x_i, \neg{x_i} \mid 1\leq i \leq n+k} \cup \set{c_i \mid 1 \leq i \leq n}\] and $\vunsat= \set{c_i \mid n+1 \leq i \leq n+k}$.
	
	Informally, we are going to make players choose valuations of the variables through their choices of transitions. The CTL formula will then ensure that the choices of transitions yield well-defined valuations, and that these valuations satisfy the SAT formula. 
	
	With that goal in mind, we define the specification as follows: 	
	
	\[\phi = (\phi_{SAT} \land \phi_{checkSat}) \lor \phi_{checkUnsat}\]
	
	\[\phi_{checkSat} = \bigwedge_{i=1}^{n}EX^{2i-2}(AX(x_i) \lor AX(\neg x_i)) \land EX \top\]
	
	\[\phi_{checkUnsat} = EX^{2n}\bigvee_{i=1}^{n}EX^{2i-2} (EX(x_i) \land EX(\neg x_i)) \land AX \bot\]
	
	and $\phi_{SAT}$ is $\psi$ where every $x_p$ has been replaced by $EX^{2p-1} x_p$ and every $\neg x_p$ replaced by  $EX^{2p-1} \neg x_p$. Recall that we assumed that $\psi$ only has negations in front of atomic propositions.
	This construction can be done in logarithmic space.
	
	The idea is that $\phi_{SAT}$ mimics $\psi$ in order to check that there exists a path in the structure obtained through the game matching a valuation satisfying $\psi$. Meanwhile, formulas $\phi_{checkSat}$ and $\phi_{checkUnsat}$ ensure that players never pick both $x_i$ or neither.
	
	Now for the formal proof, suppose there exists a valuation $\nu_1 : \set{x_1,\ldots,x_n} \to \set{\top,\bot}$ such that for every valuation $\nu_2: \set{x_{n+1}, \ldots, x_{n+k}} \to \set{\top,\bot}$, the combination of $\nu_1$ and $\nu_2$ satisfies $\psi$.

	Then let $\sigma_1(c_i) = 
	\begin{cases}
		x_i \text{ if } \nu_1(x_i) = \top\\
		\neg x_i \text{ otherwise} 
	\end{cases} \text{for } 1 \leq i \leq n$
	
	and let $\sigma_2$ be a pure strategy for $\unsat$. Clearly as $\size{\sigma_1(c_i)} = 1$ for all $i$, the resulting structure satisfies $\phi_{checkSat}$. If $\size{\sigma_2(c_i)} = 0$ for some $i$, then $\phi_{checkUnsat}$ is satisfied, thus so is $\phi$. If $\unsat$ gives every $c_i$ a successor, then there is a path from $c_{1}$ to $sink$, representing a valuation whose projection to $\set{x_1, \cdots, x_n}$ matches $\nu_1$. As a result, $\psi$ is satisfied by this valuation, thus $\phi_{SAT}$ is satisfied by the structure yielded by $\sigma_1$ and $\sigma_2$, hence so is $\phi$.
	
	Now suppose there exists a pure strategy $\sigma_1$ for $\sat$ such that for every pure strategy $\sigma_2$ for $\unsat$, $\sigma_1, \sigma_2$ yield a structure satisfying $\phi$.
	For all $1 \leq i \leq n$, if we had $\size{\sigma(c_i)} = 0$, then neither $\phi_{checkSat}$ nor $\phi_{checkUnsat}$ would be satisfied, and if we had $\size{\sigma(c_i)} > 1$, then $\phi_{checkSat}$ would not be satisfied, and $\unsat$ could win by choosing one outgoing transition for each $c_i$ he owns, thereby unsatisfying $\phi_{checkUnsat}$. As a result, $\sigma_1$ selects exactly one of $\set{x_i,\neg x_i}$ for each $i$, thus we can define $\nu_1$ the valuation such that 
	
	$\nu_1(x_i) = 
	\begin{cases}
		\top \text{ if } \sigma_1(c_i) = x_i\\
		\bot \text{ otherwise} 
	\end{cases} \text{for } 1 \leq i \leq n$
	
	Let $\nu_2 : \set{x_{n+1}, \ldots, x_{n+k}} \to \set{\top,\bot}$, we define a corresponding strategy for $\unsat$ as 
	
	$\sigma_2(c_i) = 
	\begin{cases}
		x_i \text{ if } \nu_2(c_i) = \top\\
		\neg x_i \text{ otherwise} 
	\end{cases} \text{for } n+1 \leq i \leq n+k$.
	
	As $\sigma_1, \sigma_2$ yield a structure satisfying $\phi$, either $\phi_{SAT}$ is satisfied or $\phi_{checkUnsat}$ is. Further, as in that structure every state has exactly one successor, $\phi_{checkUnsat}$ is not satisfied, thus $\phi_{SAT}$ is. 
	As a consequence, the combination of $\nu_1$ and $\nu_2$ satisfies $\psi$.
	
	We have constructed in logarithmic space a CTL formula, a modal transition system and a subset $\vsat$ of states such that $\sat$ has a pure winning strategy on $\vsat$ if and only if $\psi$ with set of existential variables $\set{x_1, \cdots,x_n}$ is in $\exists\forall$SAT.
	
	As a result the value problem corresponding to definition~\ref{2turnAsymetricCTL} is \sigmatwo-complete.
\end{proof}

\textbf{Proof of Proposition~\ref{2turnCTLusefulness}. }
The usefulness problem for two-turn CTL is \sigmathree-complete.

\begin{proof}
	Let $\modal = (S, \ap, \Delta_{must}, \Delta_{may}, init, \lambda)$ be an MTS, let $S_1, \ldots, S_n$ be a partition of $S$, let $1 \leq i \leq n$, let $\phi$ be a CTL formula. 
	
	In order to check the usefulness of $s$, we can guess a set of indices $J$ and a pure strategy $\sigma_1 : \bigcup_{j \in J \cup \set{i}} S_j \to \Delta_{may}$,
	make an adversary choose pure strategies $\sigma'_1 : \bigcup_{j \in J} S_j \to \Delta_{may}$ and $\sigma_2 : S\setminus (\bigcup_{j \in J \cup \set{i}} S_j) \to \Delta_{may}$, 
	and then guess a pure strategy $\sigma'_2 : S \setminus \bigcup_{j \in J} S_j \to \Delta_{may}$ such that the structure yielded by $\sigma_1,\sigma_2$ satisfies $\phi$ but the one yielded by $\sigma'_1,\sigma'_2$ does not.
	
	This shows that the problem is in \sigmathree.
	
	Now let us show hardness. We reduce the problem $\exists\forall\exists$SAT. Let $\exists x_1,\ldots,x_n, \forall y_{1},\ldots, y_{k}, \exists z_{1}, \ldots, z_{p} \psi$ with $\psi$ quantifier-free be a $\exists\forall\exists$SAT instance. We assume without loss of generality that all negations have been pushed to the atomic propositions.
	
	We define the MTS $\modal = (S, \ap, \Delta_{must}, \Delta_{may}, init, \lambda)$ as follows :
	
	\begin{align*}
		S = &\set{x_i, \neg x_i, c^x_i \mid 1 \leq i \leq n}\\
		\cup &\set{y_i, \neg y_i, c^y_i \mid 1 \leq i \leq k}\\
		\cup &\set{z_i, \neg z_i, c^z_i \mid 1 \leq i \leq p}\\
		 \cup &\set{x_i', \neg x_i' \mid 1 \leq i \leq n} \cup \set{win_S,win_U, s}\\
		\ap = &S\\
	\end{align*}

	\begin{align*}
	 	\Delta_{must} = &\set{(x_i,c^x_{i+1}), (\neg x_i, c^x_{i+1}) \mid 1 \leq i \leq n-1}\\
	 	\cup &\set{(y_i,c^y_{i+1}), (\neg y_i, c^y_{i+1}) \mid 1 \leq i \leq k-1}\\
	 	\cup &\set{(z_i,c^z_{i+1}), (\neg z_i, c^z_{i+1}) \mid 1 \leq i \leq p-1}\\ 
	 	\cup & \set{(x_n,c^y_1), (\neg x_n,c^y_1), (y_k,c^z_1), (\neg y_k,c^z_1)}\\
	 	\cup &\set{(x_i', x_{i+1}'),(x_i', \neg x_{i+1}') \mid 1 \leq i \leq n-1}\\ 
	 	\cup &\set{(\neg x_i', x_{i+1}'),(\neg x_i', \neg x_{i+1}') \mid 1 \leq i \leq n-1}\\ 
	 	\cup &\set{(x_i, \neg x_i'), (\neg x_i, x_i') \mid 1 \leq i \leq n}\\
	 	\cup &\set{(x_n', win_S), (\neg x_n', win_S)}\\
	 	\cup &\set{(win_{S},win_{S}), (win_{U},win_{U}), (z_p,s), (\neg z_p,s)}\\ 
	 	\cup &\set{(c^y_i, win_{U}) \mid 1 \leq i \leq k}\\ 
	 	\cup &\set{(c^z_i, win_{S}) \mid 1 \leq i \leq p}
	 \end{align*}
 
 	\begin{align*}
	 	\Delta_{may} = &\set{(c^x_i,x_{i}), (c^x_i, \neg x_{i}) \mid 1 \leq i \leq n} \\
	 	\cup &\set{(c^y_i,y_{i}), (c^y_i, \neg y_{i}) \mid 1 \leq i \leq k} \\
	 	\cup &\set{(c^z_i,z_{i}), (c^z_i, \neg z_{i}) \mid 1 \leq i \leq p} \\
	 	\cup &\set{(x_i, win_{U}),(\neg x_i, win_{U}) \mid 1 \leq i \leq n}\\
	 	\cup &\set{(x_i', win_{U}),(\neg x_i', win_{U}) \mid 1 \leq i \leq n}\\ 
	 	\cup &\set{(s,x_1'),(s, \neg x_1')}
 	\end{align*}
		
	$\lambda(t) = \set{t}$ for every state $t$, and the initial state is $init = c^x_1$.
	
	We consider the formula 
	
	\begin{align*}
		\phi = &(\neg \phi_{SAT} \land \phi_{checkSat} \land AG \neg win_{U})\\
		\lor &(EF win_S \land AG \neg win_{U})\\
		\lor &\phi_{checkUnsat}
	\end{align*}
	with 

	\begin{align*}
		&\phi_{checkSat} = \\&AG(\bigwedge_{i=1}^{k}AX(y_i) \lor AX(\neg y_i)) \land (\bigwedge_{i=1}^{k}(EX^{2n+2i-2})EX \top)
	\end{align*}
	
	\begin{align*}
		&\phi_{checkUnsat} =  \\&EF(\bigvee_{i=1}^{p}EX(z_i) \land EX(\neg z_i)) \lor \bigvee_{i=1}^{p}(EX^{2n+2k+2i-2})AX \bot
	\end{align*}

 and $\phi_{SAT}$ is $\psi$ where every $x_i, y_i, z_i$ has been replaced by respectively $EF x_i, EF y_i$ and $EF z_i$, and every $\neg x_i, \neg y_i, \neg z_i$ by respectively $AG \neg x_i, AG \neg y_i, AG \neg z_i$. This construction can be done in logarithmic space. The formulas $\phi_{checkUnsat}$ and $\phi_{checkUnsat}$ ensure that the players never allow transitions to both or neither variables from a $c_i^x, c_i^y$ or $c_i^z$ state.
	
	Suppose there exists $T$ such that $\sat$ wins with $T\cup\set{s}$ but loses with $T$.  As all $x_i, \neg x_i, x_i', \neg x_i'$ have a may transition to $win_U$, there has to be either a path from $c^x_1$ to $c^y_1$ in $T$, or a path in $T$ from $c^x_1$ to some $x_i$ or $\neg x_i$, from there a transition to some $x_i'$ or $\neg x_i'$, and a path in $T$ from there to $win_S$, otherwise $\unsat$ wins both games. In the second case, $\sat$ wins without $s$, thus we have to be in the first case. In particular for every $x_i \in T$, $x_i' \notin T$ and for every $ \neg x_i \in T$, $\neg x_i' \notin T$. 
	
	As a result, there has to be a path in $\modal$ (using may and must transitions) to all $c^x_i, c^y_i, c^z_i$ from $c^x_1$. In order for the games with $T$ and $T \cup \set{s}$ to have different winners, every $c^y_i$ has to be in $T$ (as they have a may transition to $win_U$) and similarly every $c^z_i$ has to not be in $T$. The formulas $\phi_{checkSat}$ and $\phi_{checkUnsat}$ force both players to pick exactly one outgoing transition from each $c^x_i, c^y_i, c^z_i$.
	
	Now observe that the choice of transitions from $s$ has no impact on the satisfaction of $\neg \phi_{SAT}, \phi_{checkSat}, AG \neg win_{U}$ or $\phi_{checkUnsat}$. As $\unsat$ has a winning strategy when $\sat$ only has $T$, this same strategy will ensure that $\sat$ can only win by satisfying $EF win_S \land AG \neg win_U$ in the game with $T\cup \set{s}$. In order to satisfy $EF win_S \land AG \neg win_U$, there has to be a path in $T$ from $s$ to $win_S$. As a result, at least one of $x_i', \neg x_i'$ has to be in $T$. As we have seen before, for every $x_i \in T$, $x_i' \notin T$ and for every $ \neg x_i \in T$, $\neg x_i' \notin T$, thus at most one of $x_i, \neg x_i$ can be in $T$ for all $1 \leq i \leq n$. Further, we have seen that at least one of $x_i, \neg x_i$ has to be in $T$.
	
	As a result, the set of $x_i$ in $T$ with $1 \leq i \leq n$ matches a valuation $\nu_1$ of $x_1, \cdots, x_n$. Let $\nu_2$ be a valuation of $y_{1}, \cdots, y_{k}$, suppose $\sat$ picks transitions matching $\nu_2$ from the $c^y_i$. As $\unsat$ wins the game in which $\sat$ owns only $T$, and as the satisfaction of both $\phi_{checkSat}$ and $AG\neg win_U$ is guaranteed by the strategy of $\sat$, the only possibility is that $\neg \phi_{SAT}$ is dissatisfied, which $\unsat$ can only achieve by picking transitions matching a valuation $\nu_3$ of $z_{1}, \cdots, z_{p}$ such that the combination of $\nu_1, \nu_2$ and $ \nu_3$ satisfies $\phi$. As a result, the $\exists\forall\exists$SAT instance is true. 
	
	Now for the converse, suppose there exists a valuation $\nu_1$ such that for all $\nu_2$, there exists $\nu_3$ such that their combination satisfies $\psi$.
	Let $T$ be such that $T \cap \set{x_1, \ldots, x_n, \neg{x_1}, \ldots, \neg{x_n}}$ and $T\cap \set{x_1', \ldots, x_n', \neg{x_1}', \ldots, \neg{x_n}}$ both match $\nu_1$, $T$ contains every $c^x_i$ and $c^y_i$ but does not contain $c^z_i$ for any $1 \leq i \leq p$.
	
	Let us first look at the game in which $\sat$ has states $T\cup\set{s}$.
	As one of $\set{x_i, \neg x_i}$ belongs to $T$ for all $1 \leq i \leq n$, $\sat$ can choose transitions so that there is a path from $c_1$ to $c_{n+k+1}$, and a transition from $s$ to $win_S$. 
	
	If $\unsat$ gives no outgoing transition to one of the $c_i$ with $n+k+1 \leq i \leq n+k+p$, then $\phi_{checkUnsat}$ is satisfied, thus so is $\phi$. As a result, there is a path from $c_{n+k+1}$ to either $win_{S}$ or $s$, and thus also $win_{S}$. Hence $\phi$ is satisfied in every case, $\sat$ wins that game.
	
	Now let us study the game in which $\sat$ only owns $T$. No matter which strategy $\sat$ chooses, $\unsat$ can guarantee that $EF win_S \land AG \neg win_U$ is not satisfied by allowing the transition to $win_U$ from every $x_i, \neg x_i x_i', \neg x_i'$ it owns, and not allowing any transition from $s$. This way, all paths to $win_{S}$ go through states with a transition towards $win_U$. $\unsat$ can also ensure that $\phi_{checkUnsat}$ is not satisfied by picking transitions matching some valuation from the $c_i$ he owns.
	
	Assume $\sat$ has a winning strategy to ensure that $\phi_{checkSat} \land \phi_{checkSat} \land AG \neg win_U$ is satisfied. As $\phi_{checkSat}$ and $AG \neg win_U$ are satisfied, the choices of transitions of $\sat$ from the $c_i$ have to match $\nu_1$ and a valuation $\nu_2$ of $\set{x_{n+1}, \cdots, x_{n+k}}$. There exists a valuation $\nu_3$ of $\set{x_{n+k+1}, \ldots, x_{n+k+p}}$ such that the combination of $\nu_1, \nu_2$ and $\nu_3$ satisfies $\psi$. Then if $\unsat$ chooses transitions from the $c_i$ he owns matching $\nu_3$, $\neg \phi_{SAT}$ is not satisfied by the resulting structure, contradicting the existence of a winning strategy for $\sat$.
	
	We have proven the proposition.
\end{proof}

\textbf{Proof of Theorem~\ref{2turnCTLimportance}. }
The importance computation problem associated to definition~\ref{2turnAsymetricCTL} is \sharpsigmathree-complete.

\begin{proof}
	The upper bound is easily obtained by considering the machine which guesses an ordering $\pi$ of the elements of the partition $S_1, \ldots, S_n$, computes the set $J$ of indices appearing after $s$ in $\pi$, and then calls a \sigmatwo-oracle twice to determine the winner when $\sat$ owns $\bigcup_{j \in J} S_j$ and when $\sat$ owns $\bigcup_{j \in J \cup \set{i}} S_j$.
	The number of accepting runs of the machine is then precisely the number of permutations $\pi$ matching the above condition.
	
	Now for the lower bound, we proceed in two steps. First we show that the following problem is \sharpsigmathree-complete.
	\medskip
	
	$Count\exists\forall\exists$SAT:
	\[\begin{cases}
		\text{Input:}& \text{A SAT formula } \psi \text{ over variables }\\ &\set{x_1,\ldots, x_n,y_1, \ldots,y_m, z_1, \ldots, z_r}.\\
		\text{Output: }& \text{The number of valuations of the } x_i\text{ such that}\\ 
		&\text{for all valuations of the } y_i\\
		&\text{there exists a valuation of the }z_i \text{ such that}\\
		&\text{the combination of those valuations satisfies } \psi.
	\end{cases}\]
	
	Then we show that the importance computation problem reduces to $Count\exists\forall\exists$SAT.
	
	For the first part, let $M$ be a nondeterministic Turing machine with an oracle solving a \sigmatwo-complete problem (say $\exists \forall$SAT). We can assume without loss of generality that $M$ only makes one query to the oracle, and accepts if and only if the answer is negative. 
	
	Indeed, say $M$ has to make queries $\psi_1, \ldots, \psi_k$ to the oracle, all over existential variables $y_1,\ldots, y_m$ and universal variables $z_1,\ldots,z_r$. It can nondeterministically guess the answers of the oracle and delay the verification to the end of the run. 
	
	Now let us define the order $\leq$ on valuations of the existential variables as the lexicographic order, a valuation $\nu$ being seen as the tuple $(\nu(y_1), \ldots, \nu(y_m))$ and with the convention $\bot \leq \top$. Then for the positive answers $M$ can guess a \emph{minimal} witness valuation for the $y_i$ with respect to $\leq$, negate the formula. The problem, given a $\exists\forall$SAT instance and a valuation of the existential variables, of checking whether this is the minimal valuation witnessing the satisfiability of the formula, is clearly in \conp, thus also in \pitwo. As a result $M$ can guess the minimal valuation of the $y_i$ witnessing the satisfiability of the formula, and then turn it into a $\exists \forall$SAT formula unsatisfiable if and only if the guess is correct. 
	
	Finally, in the end $M$ has to make the oracle check a disjunction of $\exists \forall$ formulas, it can rename variables in order to merge them all into one equivalent $\exists \forall$ formula, and accept if and only if the oracle rejects that formula.
	
	In all the above transformations, the number of accepting runs of the machine stays the same as the non-deterministic transitions we added (in order to guess minimal valuations witnessing satisfiability of $\exists \forall$SAT formulas) yield at most one accepting run (as the existence of such a valuation is equivalent to the existence of a single minimal one). 
	
	An adaptation of the classical construction proving that $\exists \forall$SAT is \sigmatwo-complete allows us to construct in polynomial time, given an input $w$ for $M$, a formula $\phi_1((x_i), (q_i), (s_i))$ such that the following conditions are equivalent for all valuations $\nu$ of the $x_i$, $q_i$ and  $s_i$:
	\begin{itemize}
		\item $\nu$ satisfies $\phi_1((x_i), (q_i), (s_i))$
		
		\item the $\nu(x_i)$ encode a sequence of non-deterministic choices of $M$, the $\nu(s_i)$ encode a correct run of $M$ following those choices, and the $\nu(q_i)$ encode the query made to the oracle at the end of this run
	\end{itemize}

	We can also construct in polynomial time a formula $\phi_2((y_i)_{1 \leq i \leq m}, (z_i)_{1 \leq i \leq r}, (q_i)_{1 \leq i \leq p}, (u_i)_{1 \leq i \leq k})$ simulating the oracle such that a valuation of the $q_i$ satisfies $\exists (y_i), \forall (z_i), \phi_2((y_i), (z_i), (q_i), (u_i))$ if and only if the $q_i$ encode a valid instance of $\exists \forall$SAT.
	As a result the formula 
	\begin{align*}
		&\forall (y_i), \exists (z_i), (x_i), (q_i),\\
		&\phi_1((x_i), (q_i), (s_i) \land \phi_2((y_i), (z_i), (q_i), (u_i))
	\end{align*}
	is satisfied by a valuation of the $x_i$ if and only if $M$ has a run accepting $w$ following the choices encoded by this valuation. 
	Thus the number of accepting runs of $M$ is precisely the number of valuations of the $s_i$ witnessing the validity of 
	
	\begin{align*}
		&\exists (x_i),\forall (y_i), \exists (z_i), (s_i), (q_i),\\ &\phi_1((x_i), (q_i), (s_i) \land \phi_2((y_i), (z_i), (q_i),(u_i)
	\end{align*}
	The problem $Count \exists \forall \exists$SAT is therefore \sharpsigmathree-hard.
	
	Finally, $Count \exists \forall \exists$SAT  can be reduced to the importance computation problem for 2-turn CTL using the same construction as in the proof of Proposition~\ref{2turnCTLusefulness}. Note that the $x_i, \neg x_i$ for all $n+1\leq i \leq n+k+p$, as well as $win_S, win_U$, all have no outgoing may transitions, thus have importance $0$ and thus, by a similar argument as in Lemma~\ref{UselessStatesLemma}, can be ignored in the computation of the importance. We will now only consider sets of states containing none of those. Then one can observe that the teams $T$ allowing player $\sat$ to win with $T \cup \set{s}$ but not with $T$ are exactly the teams $T$ such that 
	
	\begin{itemize}
		\item $T$ contains all the $c_i$ for $\leq  n+k$ and no other $c_i$.
		
		\item $T \cap \set{x_i, \neg x_i \mid 1 \leq i \leq n}$ and $T \cap \set{x'_i, \neg x'_i \mid 1 \leq i \leq n}$ match a same valuation $\nu$ witnessing the validity of the $\exists \forall \exists$SAT formula.
	\end{itemize}
	
	Then we have that all the teams $T$ such that $(s,T)$ is critical (and containing none of the aforementioned states with importance $0$) have the same size $M$.
	
	We obtain that the number of valuations witnessing the validity of the $\exists \forall \exists$SAT formula is $\frac{P!}{N!M! (P-M-1)!} N! I(s)$, with $N$ the number of states in the constructed MTS and $P$ the number of states minus the $x_i,\neg x_i$ for \mbox{$n+1 \leq i \leq n+k+p$}, $s$, $win_U$ and $win_S$. Hence we have a reduction from $Count\exists\forall\exists$SAT to the importance problem for 2-turn CTL.	
	As $\frac{P!}{N!M! (P-M-1)!}$ can be computed in polynomial time, the latter problem is \sharpsigmathree-complete.
\end{proof}

\end{document}